\documentclass{llncs}

\usepackage{amssymb}
\usepackage{graphicx}

\usepackage{amsmath}
\usepackage{amsfonts}

\usepackage{stmaryrd}

\usepackage{enumerate}

\usepackage{url}
 
\urldef{\mailpj}\path|petr.jancar@vsb.cz|

\usepackage{dpda-macros}

\begin{document}

\mainmatter  

\thispagestyle{plain}
\pagestyle{plain}

\title{Bisimulation Equivalence of
		First-Order Grammars\thanks{This paper extends the version
	contained in Proc. of ICALP'14.}}

	\author{Petr Jan\v{c}ar}

\institute{Dept Comp. Sci., FEI, Techn. Univ. of Ostrava
	(V\v{S}B-TUO),\\
	17. listopadu 15, 70833 Ostrava, Czech Rep.\\
\mailpj
}

\maketitle

\begin{abstract}
A decidability proof for bisimulation equivalence 
of first-order grammars
(finite sets
of labelled rules for rewriting roots of first-order terms) is
presented. The equivalence generalizes the  
DPDA
(deterministic pushdown automata) equivalence, and the result corresponds
to the result achieved by  S\'enizergues (1998, 2005) in the framework of
equational graphs, or of 
PDA with
restricted $\varepsilon$-steps. The framework
of classical first-order terms seems particularly useful for 
providing a 
proof that should be understandable for a wider audience.
We also discuss an extension to branching bisimilarity,
announced by Fu and Yin (2014).
\end{abstract}

\section{Introduction}\label{sec:intro}

Decision problems for  semantic equivalences 
have been a frequent topic in computer science. 
E.g.,
for pushdown automata (PDA) 
\emph{language equivalence} was quickly shown undecidable,
while 
the decidability in the case of deterministic PDA (DPDA) 
is 
a famous result 
by S\'enizergues~\cite{Senizergues:TCS2001}.
A finer equivalence, called \emph{bisimulation equivalence} or
\emph{bisimilarity}, has emerged as another fundamental behavioural
equivalence\,; for deterministic systems it essentially coincides with
language equivalence.
We name~\cite{BBK2} to exemplify the 
first decidability results for infinite-state systems,
and 
refer to~\cite{Srba:Roadmap:04} for a survey of a specific area.

One of the most involved results in the area~\cite{Seni05}
shows the decidability of bisimilarity
of equational graphs with finite out-degree
(or of PDA with deterministic popping
$\varepsilon$-steps); this generalizes the result for DPDA.
The recent nonelementary lower bound~\cite{BGKM12} for the problem
is, in fact, TOWER-hardness in the terminology
of~\cite{Schmitz2013}, and it holds
even for real-time PDA, i.e. PDA with no $\varepsilon$-steps. 
For the full above mentioned PDA 
the 
problem is even not primitive recursive, 
since it is Ackermann-hard~\cite{DBLP:conf/fossacs/Jancar14}.
In the deterministic case, the equivalence problem is known to be
PTIME-hard, and
has a primitive recursive upper bound shown 
by Stirling~\cite{Stir:DPDA:prim}; a finer analysis places
the problem in TOWER~\cite{DBLP:conf/fossacs/Jancar14}.
This complexity gap is just one indication that the respective
fundamental equivalence problems are far from being fully understood.
Another such indication might be the length and 
the technical nature of the so
far published proofs (including the unpublished~\cite{stirling-pda-00}).

This paper is an attempt to make a further step 
in clarifying the main decidability proof in the mentioned area.
It provides
a self-contained
decidability proof for bisimulation equivalence 
in labelled transition systems generated by \emph{first-order
grammars} (FO-grammars), which seems to be a particularly convenient
formalism. The states 
are here first-order terms over a
specified finite set of function symbols (or ``nonterminals''); the
transitions are induced by 
a finite set
of labelled rules that allow to rewrite the roots of terms.
This framework 
is equivalent
to the framework of~\cite{Seni05}; cf., e.g.,~\cite{CourcelleHandbook}
for the early references, or~\cite{JancarLICS12}
for a concrete transformation of PDA 
to FO-grammars (which is also given 
in Appendix here). 
The proof in this paper
is in principle based on the same high-level ideas as the proof
in~\cite{Seni05} but it is 
shorter and simpler; 
we do not provide a detailed comparison here.
This paper is also a (self-contained) continuation of~\cite{JancarLICS12}
where the first-order
term framework was used to give
a decidability proof in the deterministic case.

Related work is also discussed in
Section~\ref{sec:addrem}, where we address the extension of 
decidability to
branching bisimilarity, studied recently  
by Y. Fu and Q. Yin~\cite{yuxi-pdadecid-14}.

\begin{quote}
	{\small
		\emph{Remark.} Some parts are formatted as this
		remark; they contain additional details and comments.
		The aim of this paper
		is to make
		the proof easily understandable, not technically
		shortest.
	}
\end{quote}

\section{Preliminaries and Result}\label{sec:prelim}

In this section we define the basic notions and 
state the result.
Some standard definitions are restricted 
when we do not need the full generality.

By $\Nat$ we denote the 
set $\{0,1,2,\dots\}$ of nonnegative integers; we use 
$[i,j]$ to denote the set $\{i,i{+}1,\dots,j\}$.
For a set $\calA$, by $\calA^*$ we denote the set of finite
sequences of elements of $\calA$, which are also called \emph{words}
(over $\calA$).
By $|w|$ we denote the
\emph{length} of 
$w\in \calA^*$. 
By $\varepsilon$ we denote the \emph{empty sequence}
(hence $|\varepsilon|=0$).

\textbf{LTSs.}
A \emph{labelled transition system} (an LTS) 
is a tuple 
\begin{center}
$\calL=(\calS,\act,(\gt{a})_{a\in{\act}})$
\end{center}
where $\calS$ is a \emph{finite or countable}
set of \emph{states},
$\act$ is a finite 
set of \emph{actions} (or \emph{letters}),
and $\gt{a}\subseteq \calS\times\calS$ is a set of
\emph{$a$-transitions} (for each $a\in\act$). 
In fact, we only deal with  
image-finite LTSs, where 
$\calL=(\calS,\act,(\gt{a})_{a\in{\act}})$ is \emph{image-finite}
if the set $\{s'\mid s\gt{a}s'\}$ is finite for each 
pair $s\in\calS$, $a\in\act$.  We say that $\calL$ is 
a \emph{deterministic LTS} if for each pair 
$s\in\calS$, $a\in\act$ there is 
 at most one $s'$ such that $s\gt{a}s'$.

By $s\gt{w}s'$, where 
$w=a_1a_2\dots a_n\in
\act^*$,
we denote 
that there is a \emph{path}
$s=s_0\gt{a_1}s_1\gt{a_2}\cdots\gt{a_n}s_n=s'$;
if $s\gt{w}s'$, then  
$s'$ is \emph{reachable from} $s$, within $|w|$ steps.
By $s\gt{w}$ we denote that $w$ is
\emph{enabled by} $s$, i.e., $s\gt{w}s'$ for some $s'$.
A \emph{state} $s$ is \emph{dead} if 
there is no $a\in\act$ such that $s\gt{a}$.

If $\calL$ is deterministic, 
then by $s\gt{w}s'$ or $s\gt{w}$ we also denote the respective
unique path.
 
\textbf{(Stratified) bisimilarity.}
Let  $\calL=(\calS,\act,(\gt{a})_{a\in\act})$ be a given
LTS.
We say that a \emph{set} $\calB\subseteq \calS\times\calS$  
\emph{covers} 
$(s,t)\in  \calS\times\calS$ if 
\begin{itemize}
	\item		
for any $a\in\act$ and $s'\in\calS$ such that 
$s\gt{a}s'$ there is $t'\in\calS$
such that
$t\gt{a}t'$ and 
$(s',t')\in \calB$, and
\item
for any $a\in\act$ and $t'\in\calS$ such that 
$t\gt{a}t'$ there is $s'\in\calS$
such that
$s\gt{a}s'$ and 
$(s',t')\in \calB$.
\end{itemize}
We note that if $s,t$ are dead states, 
then $(s,t)$ is covered by any $\calB\subseteq \calS\times\calS$, 
in particular by $\emptyset$.
If there is an action $a\in\act$ that is enabled by precisely one
of $s,t$, then  $(s,t)$ is not covered by
any $\calB$.

For $\calB, \calB'\subseteq \calS\times\calS$
we say that $\calB'$ \emph{covers} $\calB$ if $\calB'$
covers each $(s,t)\in \calB$.
A set $\calB\subseteq \calS\times\calS$
is a \emph{bisimulation} if $\calB$ covers $\calB$.
States $s,t\in\calS$ are \emph{bisimilar},
written 
$s\sim t$,
if there is a bisimulation
$\calB$ containing $(s,t)$. 
We note the standard 
fact that 
$\sim\,\subseteq \calS\times\calS$
is the maximal
bisimulation, the union of all bisimulations.

We put $\sim_0=\calS\times\calS$. For $k\in\Nat$,
$\sim_{k+1}\subseteq\calS\times\calS$ is the set of all pairs 
covered by $\sim_{k}$. 
We easily verify that $\sim$ and $\sim_k$ are equivalence relations, and
that 
$\sim_0\,\supseteq\,
\sim_{1}\,\supseteq\,\sim_2\,\supseteq\,\cdots\cdots\supseteq \,\sim$.
For the (first infinite) ordinal $\omega$ we put 
$s\sim_\omega t$ if $s\sim_k t$ for all $k\in\Nat$; hence 
$\sim_\omega=\cap_{k\in\Nat}\sim_k$.
It is a standard fact
that
 $\cap_{k\in\Nat}\sim_k$ is a bisimulation in any 
 image-finite LTS, where we thus have
$\sim\,=\,\sim_\omega$.

\textbf{Eq-levels.} 
Given an image-finite LTS,
we attach 
the \emph{equivalence
level} (eq-level) to each pair of states:
\begin{center}
$\eqlevel(s,t)=\max\,\{k\in\Nat\cup\{\omega\}\mid s\sim_k t\}$.
\end{center}

\textbf{First-order-term LTSs informally.}
We focus on certain (image-finite) LTSs in which states are
first-order terms\,; we mean standard finite terms primarily
but 
it will turn out convenient to consider also infinite regular terms
(i.e. infinite terms with only finitely many pairwise different subterms).
The terms are built from \emph{variables}
from a fixed countable set
\begin{center}
$\var=\{x_1,x_2,x_3,\dots\}$
\end{center}
and from 
\emph{function symbols}, also called \emph{(ranked) nonterminals},
from some specified finite set $\calN$; each $A\in\calN$ has 
$\arity(A)\in\Nat$. We use $A,B,C,D$  for
nonterminals, while $E,F,\dots$ (possibly with subscripts etc.)
are reserved for terms.
\begin{figure}[t]
\centering
\includegraphics[scale=0.4]{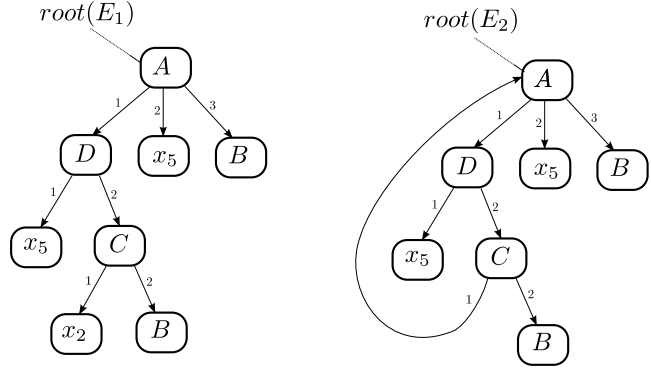}
\caption{Syntactic tree of $E_1=A(D(x_5,C(x_2,B)),x_5,B)$, and 
a graph presenting $E_2$}
\label{fig:basicterm}
\end{figure}
An example of a (standard finite) term is  
$E_1=A(D(x_5,C(x_2,B)),x_5,B)$
where the arities of $A,B,C,D$ are $3,0,2,2$, respectively.
The left-hand side of Fig.~\ref{fig:basicterm} depicts the syntactic
tree of $E_1$. (The right-hand side $E_2$ will be referred to later.)

Transitions are determined by a finite set of 
\emph{root-rewriting} rules. 
An example of a ``non-popping'' rule 
is $A(x_1,x_2,x_3)\gt{a}C(D(x_3,B),x_2)$, an
example of a ``popping'' rule is $A(x_1,x_2,x_3)\gt{b}x_1$.
Each rule induces the transitions arising by applying the same
substitution $\sigma$ to both the left-hand side (lhs) and the
right-hand side (rhs) of the rule. E.g., 
the rule  
\begin{center}
$A(x_1,x_2,x_3)\gt{a}C(D(x_3,B),x_2)$ and the substitution
$\sigma$ for which
$\sigma(x_1)=D(x_5,C(x_2,B))$,
$\sigma(x_2)=x_5$,
$\sigma(x_3)=B$ 
\end{center}
(where
$A(x_1,x_2,x_3)$ after applying $\sigma$ becomes 
$A(D(x_5,C(x_2,B)),x_5,B)$) induce the transition 
$A(D(x_5,C(x_2,B)),x_5,B)\gt{a} C(D(B,B),x_5)$ depicted
in Fig.~\ref{fig:basictransition}.
\begin{figure}[t]
\centering
\includegraphics[scale=0.4]{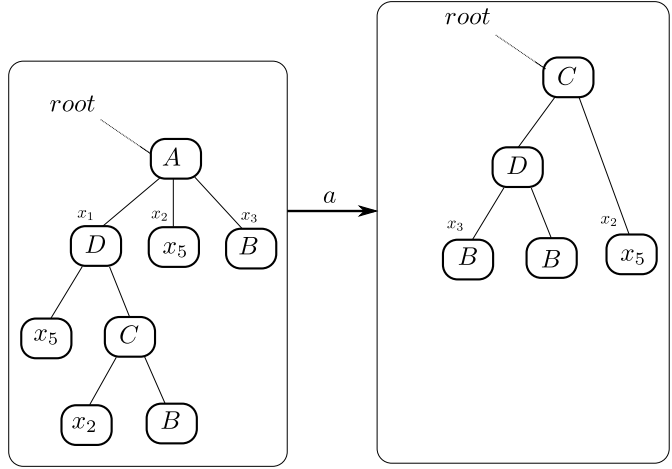}
\caption{Transition generated by 
	$A(x_1,x_2,x_3)\gt{a}C(D(x_3,B),x_2)$ and $\sigma$ (in the
text)}
\label{fig:basictransition}
\end{figure}
Fig.~\ref{fig:basictransition} depicts an $a$-transition between two states, where 
the states in our LTSs are terms.
The small symbols $x_1,x_2,x_3$ are superfluous here, they just 
  depict the original variables in the lhs and in the rhs
of the respective rule; these variables have been replaced by applying
the substitution $\sigma$.
Hence $x_3$ and $x_2$ in the target-term have been replaced by the third root-successor and by the second
root-successor of the source-term, respectively. In this concrete case
the first root-successor of the source-term ``disappears'' since 
$x_1$ does not occur in the rhs of the rule.

Another example can be given by the rule 
$A(x_1,x_2,x_3)\gt{b}x_1$ and the above $\sigma$, which induces 
the transition, or the one-step path,
$F\gt{b}H$ where $F=A(D(x_5,C(x_2,B)),x_5,B)$ and 
$H= D(x_5,C(x_2,B))$.
In this case our one-step path 
exposes a root-successor
$H$ in the
source term $F$
(the first root-successor in our case);
the path has thus ``sinked'' to a subterm in depth $1$.

\textbf{The result informally.}
We will show that there is an algorithm that computes
$\eqlevel(T_0,U_0)$ when given a finite set of root-rewriting rules
and two terms $T_0,U_0$.
In the rest of this section we formalize this statement,
 making also
some conventions about our use of (finite and infinite) 
terms and substitutions.

\textbf{Regular terms, presentation size.}
We identify terms with their syntactic trees, and denote them by
$E,F,\dots$. Thus a \emph{term $E$ over}
$\calN$ (where $\calN$ is a set of ranked nonterminals) 
is a rooted, ordered, finite or infinite tree where each node
has a label from $\calN\cup\var$; if the label of a node is $x_i\in\var$, 
then the node has no successors, and if the label is $A\in\calN$, then 
it has $m$ (immediate) successor-nodes where $m=\arity(A)$.
More precisely, a term corresponds to a set of isomorphic trees, since
two isomorphic trees represent the same term.
Each node is also the root of a \emph{subterm} of $E$, i.e., of the
subtree rooted in this node; more precisely, each concrete node 
is the root of a \emph{subterm-occurrence}, 
since a subterm corresponds to a set of isomorphic subtrees.
A subterm can thus have more (maybe infinitely
many) occurrences in $E$. Each \emph{subterm-occurrence} has
its (nesting) \emph{depth in} $E$, which is its (naturally defined) 
distance from the root of $E$; the term $E$ itself is a
subterm-occurrence with depth $0$.
E.g., in $E_1=A(D(x_5,C(x_2,B)),x_5,B)$ 
(in Fig.~\ref{fig:basicterm})
there 
is one occurrence of the
term $C(x_2,B)$, with depth $2$, and two occurrences of $x_5$, with
depths $1$ and $2$.

We also use the standard notation: a term is either $x_i$ or
$A(G_1,\dots, G_m)$; if $E=A(G_1,\dots, G_m)$, then 
$\termroot(E)=A\in\calN$, $m=\arity(A)$, and
$G_1,\dots,G_m$ are the \emph{root-successors}, i.e.,
the ordered subterm-occurrences with depth $1$.

A \emph{term} $E$ is \emph{finite} if the respective tree is finite; by
$\depth(E)$ we then mean the largest depth of a subterm-occurrence in $E$. 
For $E_1=A(D(x_5,C(x_2,B)),x_5,B)$
in Fig.~\ref{fig:basicterm}
we thus have $\depth(E_1)=3$.

A (possibly infinite) \emph{term} is \emph{regular}
if it has only finitely many subterms (though the subterms may be infinite and
can have infinitely many occurrences). 
Any regular term has a natural \emph{finite-graph
presentation} (with possible cycles).
E.g., the right-hand side of Fig.~\ref{fig:basicterm} presents a
regular term $E_2$; here the term $E_2$ itself is a subterm 
with infinitely many
occurrences (with depths $0,3,6,\dots$).
By $\pressize(E)$ (the presentation 
size of
$E$) we mean the size
of the smallest graph presentation of $E$.

\begin{quote}
{\small
We can be more precise, though the respective notions are standard.
A \emph{finite-graph
presentation} of a (regular) term over $\calN$ is a finite directed
(multi)graph, 
with a designated root,
where each node has a
label from $\calN\cup\var$; if the label of a node is $x_i\in\var$,
then the node has no outgoing arcs, and if the label is  $A\in\calN$,
then the node has $\arity(A)$ ordered outgoing arcs. 
The standard ``tree-unfolding'' of the graph is the respective term, 
which is infinite if there are cycles in the graph.
We can obviously effectively
compare if two graph presentations
represent the same term.
Given a presentation of a regular term $E$, we can thus 
compute the \emph{syntactic graph of} $E$, i.e., 
the graph whose nodes
(one-to-one) correspond to the (roots of) subterms occurring in $E$.
(E.g., the syntactic graph of $E_2$ in Fig.~\ref{fig:basicterm} arises from
the given graph presentation by merging the nodes with label $x_5$
and merging those with label $B$; but we note that if we replaced $D$ with $C$, we
\emph{could not} merge the nodes with label $C$.)
We can take the number of nodes
of the syntactic graph of $E$ as $\pressize(E)$.
}
\end{quote}
In what follows, by a ``term'' we mean a ``regular term''
if we do not say explicitly that the term is finite.
(We do not consider non-regular terms.)
We reserve symbols $E,F,G,H$, and also $T,U,V,W$, for denoting
(regular) terms. 

\textbf{Substitutions, associative composition.}
By $\trees_{\calN}$ we denote the set of all (regular) terms over
a set $\calN$ of (ranked) nonterminals.
A \emph{substitution} $\sigma$ is a mapping
\begin{center}
$\sigma:\var\rightarrow\trees_{\calN}$ whose 
\emph{support}
$\support(\sigma)=\{x_i\mid \sigma(x_i)\neq x_i\}$
\end{center}
is \emph{finite};
we reserve the symbol $\sigma$ for substitutions.
By $\range(\sigma)$ we mean the set $\{\sigma(x_i)\mid
x_i\in\support(\sigma)\}$.
By \emph{applying a substitution} $\sigma$ {to 
a term} $E$ we get the term $E\sigma$ 
that arises from $E$ by replacing each occurrence of $x_i$ with
$\sigma(x_i)$.
Hence $E=x_i$ implies
$E\sigma=x_i\,\sigma=\sigma(x_i)$; we \emph{prefer the
notation}
$x_i\sigma$ to $\sigma(x_i)$.

The 
\emph{composition of substitutions}, where
$\sigma=\sigma_1\sigma_2$ satisfies
$x_i\sigma=(x_i\sigma_1)\sigma_2$, 
can be easily verified to be
associative. We thus write simply $E\sigma_1\sigma_2$ when meaning 
$(E\sigma_1)\sigma_2$ or  $E(\sigma_1\sigma_2)$. 

\textbf{First-order grammars.}
A \emph{first-order grammar}, an \emph{FO-grammar} or
just a \emph{grammar} for short, is a tuple
$\calG=(\calN,\act,\calR)$ where 
$\calN$
is a finite set of 
ranked \emph{nonterminals}, viewed as function symbols with
arities, $\act$
is a finite set of \emph{actions} (or letters), 
and $\calR$
is
a finite set of 
\emph{rules} of the form
\begin{center}
$A(x_1,x_2,\dots, x_m)\gt{a} E$
\end{center}
where $A\in \calN$, $\arity(A)=m$, 
$a\in\act$,
and $E$ is a
\emph{finite}  
term over $\calN$ 
\emph{in which each occurring variable 
is
from the set} $\{x_1,x_2,\dots,x_m\}$.

\textbf{Rule-based and action-based LTSs generated by grammars.}
Given $\calG=(\calN,\act,\calR)$, 
by $\calL^{\ltsrul}_{\calG}=(\trees_{\calN},\calR,(\gt{r})_{r\in\calR})$
we denote the (\emph{rule based}) LTS
where each rule $r$ of the form
$A(x_1,x_2,\dots, x_m)\gt{a} E$ 
induces
	$(A(x_1,\dots, x_m))\sigma\gt{r}E\sigma$
for any substitution $\sigma$.
\\
(Hence also $A(x_1,\dots, x_m)\gt{r}E$, due to $\sigma$
with $\support(\sigma)=\emptyset$.)

\begin{quote}
{\small
Speaking in an informal 
``operational'' manner,
we can apply a rule $r$ of the form  $A(x_1,\dots, x_m)\gt{r}E$
to (a graph-presentation of)
$F$ iff 
$\termroot(F)=A$.
If so, and $E=x_j$, then the
target of the $j$-th outgoing arc of the root of $F$
(which might be the root itself in the case of a loop) is the root of 
$H$ where $F\gt{r}H$.
If $E\not\in\var$, we get $H$ (for which $F\gt{r}H$)
by adding (a fresh copy of) $E$
to $F$ where the root of $E$ becomes the root of the arising $H$; 
to finish the construction of $H$, 
each
arc in
$E$ leading to a node labelled with $x_j$ is redirected to 
the $j$-th root-successor in $F$. 
Hence the variables $x_1,\dots,x_m$ in the rules serve just as
``place-holders'' for root-successors (in the source term of a
transition); recall again Fig.~\ref{fig:basictransition}.
}
\end{quote}
The LTS $\calL^{\ltsrul}_{\calG}$ is deterministic, since for each $F$
and $r$ there is at most one $H$ such that $F\gt{r}H$.
Hence $F\gt{w}$, for $w\in\calR^*$, refers to a unique path in
$\calL^\ltsrul_\calG$ (which is later technically convenient).

\begin{quote}
{\small
We stress explicitly that \emph{transitions cannot add variables},
i.e., $F\gt{w}H$ implies that
each variable occurring in $H$ also occurs in $F$ (though not vice
versa in general; recall that the first root-successor 
``disappeared'' by the transition in Fig.~\ref{fig:basictransition},
which also caused that the subterm $x_2$ ``disappeared'').
We also note that  $F\gt{w}H$ implies $F\sigma\gt{w}H\sigma$ for any
substitution $\sigma$; this follows from the fact that $x_i$ are dead,
 not enabling any action.
\\
Finally we observe that our stipulation that the 
right-hand sides (rhs) $E$ in the grammar-rules
$A(x_1,\dots,x_m)\gt{a}E$
are finite implies
that
\emph{all terms reachable from a finite term} are \emph{finite}.
(It turns out technically convenient to have the rhs finite 
while including regular terms into our LTSs.)
}
\end{quote}
By the \emph{action-based} LTS, related to a grammar
$\calG=(\calN,\act,\calR)$,  
we mean the LTS  $\calL^{\ltsact}_{\calG}=(\trees_{\calN},\act,(\gt{a})_{a\in\act})$
where each rule $A(x_1,\dots, x_m)\gt{a}E$
induces  
\begin{center}
$(A(x_1,\dots, x_m))\sigma\gt{a}E\sigma$ 
\end{center}
for any 
substitution 
$\sigma$.

Hence $F\gt{w}H$ in $\calL^\ltsrul_\calG$ implies 
 $F\gt{\lab(w)}H$ in $\calL^\ltsact_\calG$, where
 $\lab(w)$ is the naturally defined \emph{action-image} of $w$:
the homomorphism $\lab:\calR^*\rightarrow\act^*$ is defined by putting
$\lab(r)=a$ for any rule $r$ of the form $A(x_1,\dots,x_m)\gt{a}E$.

We note that $\calL^{\ltsact}_{\calG}$ is image-finite,  
and nondeterministic in general.
In fact, we still \emph{complete the definition of} 
$\calL^{\ltsact}_{\calG}$ \emph{by stipulating that}
\begin{center}
\emph{no $\calB\subseteq\trees_\calN\times\trees_\calN$ covers $(x_i,H)$ or $(H,x_i)$ 
when
$H\neq x_i$.}
\end{center}
We thus have that
\begin{center}
$x_i\neq H$ implies $x_i\not\sim_1 H$, i.e., 
$\eqlevel(x_i,H)=0$.
\end{center}
In particular we have
$x_i\not\sim_1 x_j$ for $i\neq j$.
Technically we think of each used variable $x\in\var$ 
as being equipped with its unique action $a_{x}$ and with the
transition $x\gt{a_{x}}x$ in $\calL^\ltsact_\calG$;
this entails that $x_i\not\sim_1 H$
for $H\neq x_i$
without any special stipulation. 

\smallskip

\emph{Convention.}
Whenever we consider $F\gt{w}H$ in 
$\calL^\ltsact_\calG$, we tacitly assume that no 
special transitions $x\gt{a_x}x$ are involved.
Hence $F\gt{w}H$ implies $F\sigma\gt{w}H\sigma$ for any substitution
$\sigma$. We also stipulate that $\emptyset$
covers $(x_i,x_i)$, thus avoiding superfluous technicalities.

\begin{quote}
{\small
The stipulation $x_i\not\sim_1 H$ for $H\neq x_i$
reflects the fact that $x_i\neq H$ implies that
$x_i\sigma\not\sim_1 H\sigma$ for
some $\sigma$, unless the underlying grammar $\calG$
is trivial.
The special transitions $x\gt{a_x}x$ are just one technical possibility 
how to reflect this fact in $\calL^\ltsact_\calG$ smoothly.
}
\end{quote}
In what follows we refer to the action-based LTSs
$\calL^{\ltsact}_{\calG}$, 
if we do not say explicitly that we have $\calL^{\ltsrul}_{\calG}$ in
mind. 

\begin{theorem}\label{th:bisdecid}
There is an algorithm that, given an FO-grammar
$\calG=(\calN,\act,\calR)$ and
$T_0,U_0\in\trees(\calN)$,
computes $\eqlevel(T_0,U_0)$ 
in  $\calL^{\ltsact}_{\calG}$.
\end{theorem}

\section{Proof of Theorem~\ref{th:bisdecid}}

We note that deciding $\sim_0$ is trivial, since 
$T\sim_0 U$ holds for all $T,U$.  
When having a procedure deciding $\sim_k$, we can 
easily construct a procedure deciding $\sim_{k+1}$;
this follows from the facts that 
$T\sim_{k+1}U$ iff $(T,U)$ is covered by
$\sim_k$, and that
for any $V$ we can
construct all (finitely many) pairs $(a,V')$ such that $V\gt{a}V'$.
We thus get a part of 
Theorem~\ref{th:bisdecid}:

\begin{proposition}\label{prop:negatcase}
There is an algorithm that, given $\calG$ and $T_0,U_0$,
outputs $\eqlevel(T_0,U_0)$ if $T_0\not\sim U_0$, and does not halt if
$T_0\sim U_0$.
\end{proposition}
We need to modify the algorithm so that it recognizes the case 
$T_0\sim U_0$ in
finite time.
As a convenient  tool
we introduce a round-based game between
Prover(she) and Refuter(he); 
the game is more involved than the standard
bisimulation game. 
We start with a simple first version of the game, and then we
 enhance it stepwise.
Refuter will be always able to force his win in finite time 
if the terms in the initial pair $(T_0,U_0)$ are non-equivalent.
Prover will be always able to avoid losing if $T_0\sim U_0$,
but only in the
last game-version she will be able to force her win in finite
time. Since Prover's winning strategy for $T_0\sim U_0$
in the last game-version
will be finitely presentable and
effectively verifiable, a proof of  Theorem~\ref{th:bisdecid} will be
finished.
Before the first game-version we 
observe some simple standard facts related to
(stratified) bisimulation equivalence.

\textbf{Expansions.}
Assume an LTS $\calL=(\calS,\act,(\gt{a})_{a\in\act})$.
By
\begin{center}
$\calB\iscov\calB'$, where $\calB,\calB'\subseteq \calS\times\calS$,
\end{center}
we denote that $\calB'$ is a \emph{minimal expansion for $\calB$},
i.e., $\calB'$ 
covers 
$\calB$ and no proper subset of $\calB'$ covers $\calB$;
this also implies
that for each $(s',t')\in\calB'$ there is $(s,t)\in\calB$
such that $s\gt{a}s'$ and $t\gt{a}t'$ for some $a\in\act$.
We note that $\emptyset\iscov\emptyset$, and
if $s,t$ are dead (not enabling any action), then
$\{(s,t)\}\iscov\emptyset$.

For any $k\in\Nat$ we have
$k<\omega$ and we stipulate  $\omega-k=\omega+k=\omega$.
We also 
stipulate $\min\emptyset =\omega$, and define
$\leasteqlev(\calB)=\min\{\eqlevel(s,t)\mid (s,t)\in\calB\}$.

\begin{proposition}\label{prop:simplecovering}\hfill\\
(1) If  $\leasteqlev(\calB)=0$ then there is no $\calB'$ 
such that  $\calB\iscov\calB'$.
\\
(2) If $\calB\iscov\calB'$ and  $\leasteqlev(\calB)<\omega$, then 
$\leasteqlev(\calB)>\leasteqlev(\calB')$.
\\
(3) If $\leasteqlev(\calB)>0$ then there is $\calB'$ such that 
$\calB\iscov\calB'$ and $\leasteqlev(\calB')\geq
\leasteqlev(\calB)-1$. 
(In particular, if $\calB\subseteq\,\sim$ then $\calB\iscov\calB'$ for
some $\calB'\subseteq\,\sim$.)
\\
(4) For $k\in\Nat$ we have $s\sim_k t$ iff there is a sequence
$\{(s,t)\}\iscov \calB_1\iscov\calB_2\iscov\cdots\iscov \calB_k$.
\end{proposition}
\textbf{Prover-Refuter game (first version).}
A play starts with a grammar
$\calG=(\calN,\act,\calR)$ and an \emph{initial pair}
$(T_0,U_0)$ of terms.
For $i=0,1,2\dots$, the \emph{$(i{+}1)$-th round} of the play
starts with some specified pair $(T_i,U_i)$ and proceeds
as follows:

\begin{enumerate}
\item
Prover chooses $k>0$ and some 
$\calB_j\subseteq\trees_\calN\times\trees_\calN$ for $j=1,2\dots,k$
and shows that 
$\calB_0\iscov\calB_1\iscov\calB_2\iscov\ldots\iscov\calB_k$  
where $\calB_0=\{(T_i,U_i)\}$.
\\
If this is impossible 
(i.e., if $T_i\not\sim_1 U_i$), then Refuter
wins.
\item
Refuter chooses a pair $(T'_i,U'_i)$ in 
$\calB_k\smallsetminus \bigcup_{j=0}^{k-1}\calB_j$. 
If this is impossible, i.e. 
if $\calB_k\subseteq  \bigcup_{j=0}^{k-1}\calB_j$ (which includes the
case $\calB_k=\emptyset$),
then  Prover wins.
(In this case  $T_i\sim U_i$, 
since  $T_i\not\sim U_i$ implies that 
$\calB_k\not\subseteq \bigcup_{j=0}^{k-1}\calB_j$,
by Prop.~\ref{prop:simplecovering}(2).)
\item
The pair $(T_{i+1},U_{i+1})=(T'_i,U'_i)$ is taken for starting 
the $(i{+}2)$-th round.
\end{enumerate}
\begin{figure}
\centering
\includegraphics[scale=0.52]{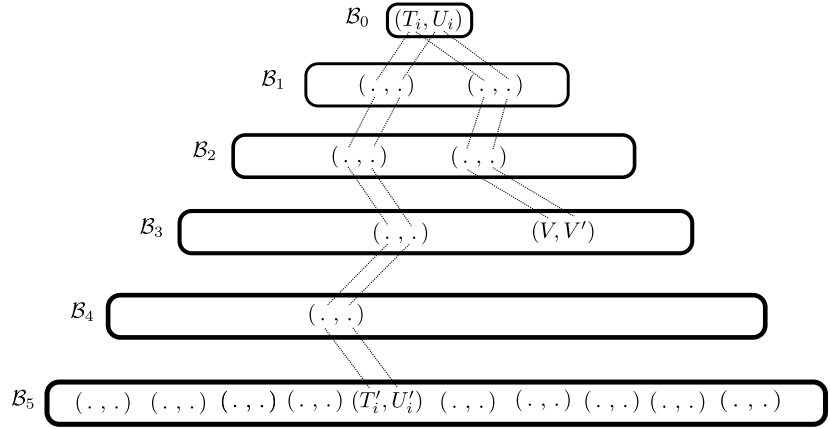}
\caption{Illustration of a game-round (where Prover has chosen $k=5$)}
\label{fig:roundsketch}
\end{figure}
Fig.~\ref{fig:roundsketch} illustrates an $(i{+}1)$-th round (with
$k{=}5$).
Note that 
our requirement $\calB_0\iscov\calB_1\iscov\cdots\iscov\calB_k$
entails that there are $u_1,u_2\in\calR^*$ such that $|u_1|=|u_2|=k$, $\lab(u_1)=\lab(u_2)$, and $T_i\gt{u_1}T'_i$, 
$U_i\gt{u_2}U'_i$ (in $\calL^{\ltsrul}_\calG$), as sketched in the
figure.
\begin{quote}
	{\small	
The reversals of $u_1,u_2$ can be found by the ``bottom-up
approach'', starting from
 $(T'_i,U'_i)$ and going up to $(T_i,U_i)$.
This follows from recalling that $\calB\iscov \calB'$ implies that 
for each $(G',H')\in\calB'$ there is $(G,H)\in\calB$ such that 
$G\gt{a}G'$ and $H\gt{a}H'$ for some $a\in\act$.
}
\end{quote}
We also note that, e.g., 
for each $V$ such that $T_i\gt{v}V$ in $\calL^{\ltsrul}_\calG$
where $|v|=j\leq k$ there is $v'\in\calR^*$ and $V'$ such that 
$|v'|=|v|$, $\lab(v)=\lab(v')$, $U_i\gt{v'}V'$,
and $(V,V')\in\calB_j$.
\begin{quote}
	{\small	
This is also depicted in Fig.~\ref{fig:roundsketch}. When looking for
$U_i\gt{v'}V'$, we now use the
``top-down approach'' driven by $T_i\gt{v}V$.
}
\end{quote}
We say that \emph{Refuter} uses the \emph{least-eqlevel strategy}, if
he always chooses $(T'_i,U'_i)$ so that 
$\eqlevel(T'_i,U'_i)=\leasteqlev(\calB_k)$; in this case 
$\eqlevel(T'_i,U'_i)<
\leasteqlev(\bigcup_{j=0}^{k-1}\calB_j)$,
and thus $\eqlevel(T'_i,U'_i)< \eqlevel(T_i,U_i)-(k{-}1)$,
unless $T\sim U$ for all 
$(T,U)\in \bigcup_{j=0}^{k}\calB_j$.
We easily observe the following facts.

\begin{proposition}\label{prop:firstgameanal}
	Let $\eqlevel(T_0,U_0)=e\in\Nat\cup\{\omega\}$.
\\
1. If $e<\omega$, then Refuter  wins within $e{+}1$ rounds by 
the least-eqlevel strategy. 
\\
2. Prover can guarantee that she will not lose within $e$ rounds.
\end{proposition}

\smallskip

\textbf{Prover's additional tool.}
A challenge is to add sound possibilities for Prover to enable her to
force her win in finite time if $T_0\sim U_0$.
We allow
Prover to claim a win 
when she can (soundly) demonstrate, in some $(i{+}1)$-th round, 
that either Refuter has not used the least-eqlevel strategy 
or $T_0\sim U_0$.
This new abstract rule does not
change Prop.~\ref{prop:firstgameanal}.
A simple instance 
is a \emph{repeat}: if 
$\{T_{i},U_{i}\}= \{T_{j},U_{j}\}$ for some $j<i$,
then Prover can claim her win. 
(Equality $T_{i}=U_{i}$ is another trivial example.)

We thus further assume that Prover wins when a repeat appears,
and we look at more involved options
enabling her to ``balance'', i.e., to replace
$T'_i,U'_i$ (in the point $3$ of a game-round)
with  $T_{i+1},U_{i+1}$ that are 
``closer'' to each other, while 
keeping  $\eqlevel(T_{i+1},U_{i+1})=\eqlevel(T'_i,U'_i)$
when Refuter
uses the least eq-level strategy.
Before formulating the second version of the game,
we clarify the crucial underlying facts. First a trivial one:

\begin{proposition}\label{prop:basicreplace}
Assume an LTS $\calL$.
	If $\eqlevel(s,t){=}k$ and $\eqlevel(s,s')> k$,
then $\eqlevel(s',t)=k$
(since $s'\sim_k s\sim_k t$ and $s'\sim_{k+1} s\not\sim_{k+1}t$). 
\end{proposition}

\textbf{Congruence, the crux of balancing.}
We assume a given grammar $\calG=(\calN,\act,\calR)$.
For substitutions
$\sigma,\sigma'$ and $k\in\Nat\cup\{\omega\}$ we put
\begin{center}
$\sigma\sim_k \sigma'$ 
if 
$x_i\sigma\sim_k x_i\sigma'$ for
each $x_i\in\var$.
\end{center}
We also put
$\eqlevel(\sigma,\sigma')=\max\,\{\,k\in\Nat\cup\{\omega\}\mid 
\sigma\sim_k\sigma'\}$; hence
$\eqlevel(\sigma,\sigma')=
\min \,\{\,\eqlevel(x_i\sigma, x_i\sigma')\mid x_i\in\var\}$.
We now note that $\sim_k$ and $\sim=\sim_\omega$ are congruences:

\begin{proposition}\label{prop:congruence}\hfill\\
(1) If $E\sim_k F$, then $E\sigma\sim_k F\sigma$;
hence $\eqlevel(E,F)\leq \eqlevel(E\sigma,F\sigma)$.
\\
(2) If $\sigma\sim_k \sigma'$, 
then $E\sigma\sim_k E\sigma'$; 
hence $\eqlevel(\sigma,\sigma')\leq \eqlevel(E\sigma,E\sigma')$.
\\
Moreover, if $\sigma\sim_k \sigma'$ and
$E\not\in\var$ (i.e., $\termroot(E)\in\calN$),
then $E\sigma\sim_{k+1} E\sigma'$.
\end{proposition}	

\begin{proof}
(1) Suppose $\{(E,F)\}=\calB_0\iscov{\calB_1}\iscov{\calB_2}\cdots
\iscov{\calB_k}$; note that any pair $(x_i,H)$, or $(H,x_i)$,
in $\bigcup_{j=0}^{k-1}\calB_j$
must satisfy $H=x_i$ (since otherwise it cannot be covered by our
definition of $\calL^\ltsact_\calG$).
For each $j\in[0,k]$ we put $\calB'_j=\{(G\sigma,H\sigma)\mid
(G,H)\in\calB_j\}$. We almost get
$\{(E\sigma,F\sigma)\}=\calB'_0\iscov{\calB'_1}\iscov{\calB'_2}\iscov
\cdots
\iscov{\calB'_k}$; just for the cases 
$(x_i,x_i)\in\calB_j$, $j\in[0,k{-}1]$, where 
 $(x_i\sigma, x_i\sigma)$ is not
 covered by $\calB'_{j+1}$, we complete  
 $\calB'_{j+1}, \calB'_{j+2}, \dots, \calB'_k$
 with some pairs of identical
 terms, using the fact that 
 $\{(E,E)\}\iscov \{(E',E')\mid E\gt{a}E'$ for some $a\in\act\}$).

(2) Suppose $\sigma\sim_k \sigma'$, and take 
 $\{(E,E)\}=\calB_0\iscov{\calB_1}\iscov{\calB_2}\cdots
 \iscov{\calB_k}$ where $\bigcup_{j=0}^{k}\calB_j\subseteq\,
 \{(F,F)\mid F\in\trees_\calN\}$. 
Let $\calB'_j=\{(G\sigma,G\sigma')\mid (G,G)\in\calB_j\}$.
For the cases $(x_i,x_i)\in\calB_j$, and thus 
$(x_i\sigma,x_i\sigma')\in\calB'_j$, we complete 
 $\calB'_{j+1}, \calB'_{j+2}, \dots\calB'_k$ accordingly, using the fact that 
$x_i\sigma\sim_k x_i\sigma'$. If $E\not\in\var$, then this procedure
is valid even when we 
start with
$\{(E,E)\}=\calB_0\iscov{\calB_1}\iscov{\calB_2}\cdots
\iscov{\calB_{k+1}}$.
\qed
\end{proof}

\begin{quote}
{\small
\emph{Remark.}	
The compositionality induced by the congruence properties leads
naturally to considering the following modification of our game, based
on decompositions.
Prover is always allowed to 
``decompose'' $(T_i,U_i)$, 
i.e., to present
some finite set
$\calB$
of pairs of terms that are in some sense smaller
than $(T_i,U_i)$, if it is guaranteed that
$\eqlevel(T_i,U_i)\geq \leasteqlev(\calB)$; 
Refuter then chooses a pair from $\calB$ for continuing.
If our measure of size satisfies 
that there are only finitely many pairs with the size 
that is smaller than or equal to the size of any given $(T,U)$, then 
Refuter's least-eqlevel strategy still wins if the initial terms are
non-equivalent. On the other hand, if there is a bound such that each
pair $T\sim U$ that is bigger than the bound is decomposable via a set
$\calB\subseteq\sim$, then we have a required algorithm: 
if Prover keeps decomposing large pairs via subsets of $\sim$, then we
get a repeat eventually. 

This is a basis of decision algorithms for so called BPA processes,
which can be viewed as being generated by FO-grammars where all
nonterminals have arity $1$ (or $0$).
We can refer, e.g., to the
papers~\cite{DBLP:journals/iandc/ChristensenHS95,DBLP:conf/mfcs/BurkartCS95,Jan12b}
for details. 
However, in our more general case such a straightforward approach
does not seem clear. We follow a more involved way, based on a balancing
strategy that makes the component-terms in $(T_i,U_i)$ 
``close to each other''. 
We note that 
balancing strategies, in different frameworks, 
were used by 
S\'enizergues~\cite{Senizergues:TCS2001,Senizergues:TCS2002simple,Seni05}
and then Stirling~\cite{Stirling:TCS2001,Stir:DPDA:prim,stirling-pda-00}.
In fact, we will also discuss 
a bit of decomposition later, in Section~\ref{sec:addrem}.
}
\end{quote}
We now illustrate how Prover can use already the simple fact captured by
Prop.~\ref{prop:basicreplace}.
Suppose the $(i{+}1)$-th round 
starts with $(T_i, U_i)$ and Refuter chooses 
$(T'_i,U'_i)$ in $\calB_k$ (we refer to the notation in the game
definition, and to Fig.~\ref{fig:roundsketch}).
We thus have $T_i\gt{u_1}T'_i$, $U_i\gt{u_2}U'_i$ in $\calL^{\ltsrul}_\calG$,
for some $u_1,
u_2\in\calR^*$, where  $|u_1| = |u_2|= k$ (and $\lab(u_1)=\lab(u_2)$).

Suppose that $T_i\gt{u_1}T'_i$ is \emph{not a shortest path} from 
$T_i$ to $T'_i$ (in $\calL^{\ltsrul}_\calG$);
then we have $T_i\gt{v_1}T'_i$ for some $v_1\in\calR^*$ 
where $|v_1|<|u_1|$. As we already observed,  we then must also have 
$U_i\gt{v_2}U''$ for some $U''$ and some $v_2\in\calR^*$ such that 
$|v_1|=|v_2|$ 
and $(T'_i,U'')\in\bigcup_{j=0}^{k-1}\calB_j$.
(In Fig.~\ref{fig:roundsketch} we would have a respective pair
$(V,V')=(T'_i,U'')$.)
We thus have  $\eqlevel(T'_i,U'')>\eqlevel(T'_i,U'_i)$ when $T'_i\not\sim
U'_i$ and Refuter uses the least-eqlevel strategy.

Therefore Refuter ``cannot protest'' when Prover puts
$(T_{i+1},U_{i+1})=(U'',U'_i)$ instead of 
$(T_{i+1},U_{i+1})=(T'_i,U'_i)$, since
$\eqlevel(T'_i,U'_i)=\eqlevel(U'',U'_i)$ if Refuter uses the
least-eqlevel strategy.
We note that $U'',U'_i$ are close to each other in the
sense that they are both reachable within $k$ steps from one
``pivot term'', namely $U_i$. 
We have $U_i\close_{k} (U'',U'_i)$, where generally we define 
\begin{equation}\label{eq:relclose}
\begin{minipage}{0.8\textwidth}
$W\close_{k}T \iffdef W\gt{v}T$ for some  
$v \text{ of length at most } k$, and
\\
$W\close_{k}(T,U)\iffdef W\close_{k}T$ and $W\close_{k}U$.
\end{minipage}
\end{equation}
In the second game-version below
we use the congruence properties
to enable
Prover to replace $(T'_i,U'_i)$ with ``closer''
$(T_{i+1},U_{i+1})$ even in some cases where $T_i\gt{u_1}T'_i$ is
a shortest path from $T_i$ to $T'_i$ and 
$U_i\gt{u_2}U'_i$ is a shortest path from $U_i$ to $U'_i$.

\smallskip

\textbf{Prover-Refuter game (second version).}
The only change w.r.t. the first game-version is in the
 point $3$:
\begin{enumerate}[3.]
\item
	Prover creates $(T_{i+1},U_{i+1})$ for the start of the 
	$(i{+}2)$-th round:

	Either she puts $(T_{i+1},U_{i+1})=(T'_i,U'_i)$,
	thus making \emph{no change}, 
	or she can use one of the following options if available:
	\begin{enumerate}[i/]
		\item \emph{Left-balancing}:
Prover presents $T'_i$ as $G\sigma$ for some \emph{finite
term} $G$ and some substitution $\sigma$,
where for each
$V\in\range(\sigma)$ she finds $V'$ such that 
$(V, V')\in \bigcup_{j=0}^{k-1}\calB_j$.
She defines $\sigma'$ with 
$\support(\sigma')=\support(\sigma)$
as follows: if $\sigma(x_\ell)=V$, then
$\sigma'(x_\ell)=V'$, where $(V,V')$ is an above found pair.
Finally she  puts 
$(T_{i+1},U_{i+1})=(G\sigma',U'_i)$.
\item
\emph{Right-balancing}:	
Symmetrically, Prover presents $U'_i$ as $G\sigma$, 
finds all appropriate pairs
$(V',V)$ in $\bigcup_{j=0}^{k-1}\calB_j$, 
and puts $(T_{i+1},U_{i+1})=(T'_i,G\sigma')$. 
\end{enumerate}
\end{enumerate}
Our previous illustration, where $T_i\gt{u_i}T'_i$ 
was not a shortest path from $T_i$ to $T'_i$,
was a special case:
we had $T'_i=G\sigma$ where $G=x_1$, $\support(\sigma)=\{x_1\}$ and 
$x_1\sigma=T'_i$, and we replaced $G\sigma$ with 
$G\sigma'$ where $x_1\sigma'=U''$ (and thus 
$(x_i\sigma,x_i\sigma')\in\bigcup_{j=0}^{k-1}\calB_j$
for all $x_i\in\support(\sigma)=\support(\sigma')=\{x_1\}$).

Informally speaking, in the second game-version
Prover might replace the whole $T'_i$ with some $U''$ that
is ``shortly reachable from the pivot'' (if possible),
but she can also replace just
``small-depth'' subterms $V$ of $T'_i$ with (sub)terms $V'$ that are
``shortly reachable from the pivot''; in the latter case some
``\emph{special finite head}'' $G$ of $T'_i$ remains. (The case of
right-balancings is symmetric.)

A left-balancing (with the pivot $W=U_i$ and the bal-result 
$(G\sigma',U'_i)$)
is also depicted in the upper part of
Fig.~\ref{fig:tworounds}. (The lower part will be discussed later.)

\begin{figure}
\centering
\includegraphics[scale=0.52]{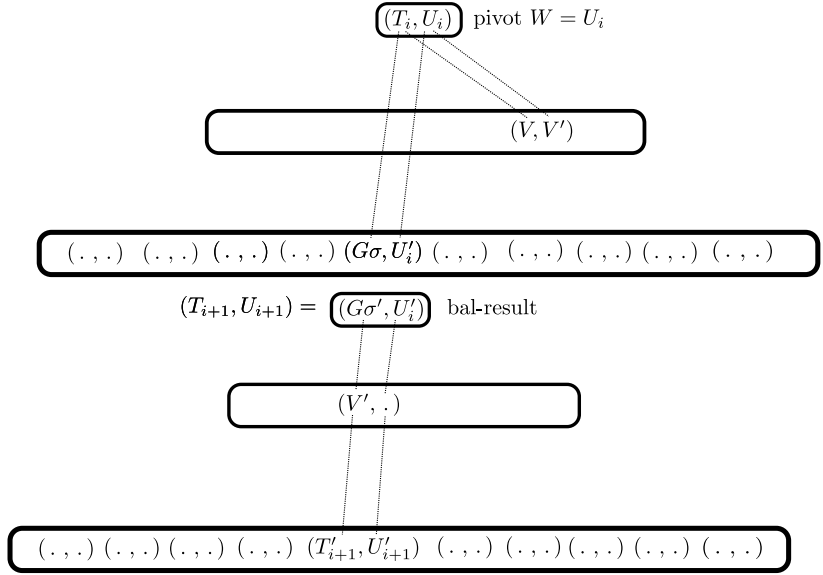}
\caption{Two consecutive rounds, with a left-balancing in the first
one}
\label{fig:tworounds}
\end{figure}

\smallskip
We can easily verify that Prop.~\ref{prop:firstgameanal}
holds also for the second game-version. The crucial point is that
$\eqlevel(T_{i+1},U_{i+1})=\eqlevel(T'_i,U'_i)$ when Refuter uses the
least-eqlevel strategy 
(this is based on Prop.~\ref{prop:congruence}(2) and 
Prop.~\ref{prop:basicreplace}).

\textbf{Bal-results are close to pivots.}
When doing a left-balancing, 
replacing 
$(T'_i,U'_i)=(G\sigma,U'_i)$ with 
the \emph{bal-result}
$(T_{i+1},U_{i+1})=(G\sigma',U'_i)$,
we might not have $U_i\close_{k} (T_{i+1},U_{i+1})$
for the \emph{pivot} $U_i$,
 but we surely 
have
$U_i\closelh{d}_{k} (T_{i+1},U_{i+1})$, for $d=\depth(G)$,
where we generally extend the notation from~(\ref{eq:relclose})
as follows: 
\begin{equation}\label{eq:closelh}
\begin{minipage}{0.8\textwidth}
$	
W\closelh{d}_{k} (T,U) \iffdef 
\text{ there is 
a finite term }
	G \text { and a substitution } 
	\sigma$ 
\\	
such that 
$T=G\sigma, \depth(G)\leq d$,  $W\close_{k}U$,
 and 
\\
$W\close_{k}V$
for all $V\in\range(\sigma)$;
\end{minipage}
\end{equation}
the symbol \textsc{L} signals that we allow a special head
in the \emph{left}-hand component (here with the height at most $d$).
Symmetrically 
we define $W\closerh{d}_{k} (T,U)$, 
where \textsc{R} refers to the right-hand component.

\smallskip

\textbf{Two remaining steps 
in the proof of Theorem~\ref{th:bisdecid}.}
 We first give an informal sketch, which is then formalized.
 We recall that a play of the Prover-Refuter game gives
rise to a sequence 
$(T_0,U_0),(T_1,U_1), (T_2,U_2),\dots$
of pairs of terms that are the starting pairs for 
the rounds $1,2,3,\dots$, respectively.

In the first of the remaining proof steps
(captured by Lemma~\ref{lem:forcingngseq})
we show that in the case $T_0\sim U_0$ Prover can force a certain
potentially 
infinite $(n,g)$-subsequence of $(T_1,U_1), (T_2,U_2),\dots$, by a
simple balancing strategy.

In the second step
(Lemma~\ref{lem:realboundng}) we bound
the lengths of \emph{eqlevel-decreasing} $(n,g)$-sequences.
It turns out that Prover can compute a respective bound
$\ell_{n,g}\in\Nat$ on condition that she guesses correctly the pairs
of equivalent terms up to a certain (large) presentation size.

In the final game-version Prover wins if the length of a created
$(n,g)$-sequence exceeds $\ell_{n,g}$, but
Refuter's least-eqlevel strategy will be
still winning if Prover does not guess correctly when computing 
$\ell_{n,g}$.

\textbf{Eqlevel-decreasing sequences, and $(n,g)$-sequences.}
A \emph{sequence}  
$(V_1,V'_1), (V_2,V'_2),(V_3,V'_3), \dots$
of pairs of (regular) terms is \emph{eqlevel-decreasing} if 
$\omega>\eqlevel(V_1,V'_1)>\eqlevel(V_2,V'_2)>\cdots$. 
In this case the sequence must be finite, and our requirement
$V_1\not\sim V'_1$ implies that its length is
bounded by $1{+}\eqlevel(V_1,V'_1)$. 

Given a pair $(n,g)$ where $n\in\Nat$ and
$g:\Nat_+\rightarrow\Nat_+$ is a nondecreasing function
(where $\Nat_+=\{1,2,\dots\}$), 
a (finite or infinite) sequence of pairs of terms
is an \emph{$(n,g)$-sequence} if it can be presented as
\begin{center}
$(E_1\sigma,F_1\sigma), (E_2\sigma,F_2\sigma),(E_3\sigma,F_3\sigma),
\dots$
\end{center}
for a substitution $\sigma$
with $|\support(\sigma)|\leq n$, where 
$\pressize(E_j,F_j)\leq g(j)$ for $j=1,2,\dots$. 
(We put $\pressize(E,F)=\pressize(E)+\pressize(F)$, say.)
Thus the growth of the (regular) ``head-terms'' $E_j,F_j$ is bounded by
the function $g$, while at most $n$ fixed 
``tail-subterms'' (of unrestricted
size)
suffice for this presentation.

\textbf{Prover can force an $(n,g)$-subsequence by a balancing
strategy.}
We aim to prove Lemma~\ref{lem:forcingngseq}; the proof is the most
technical part of the paper, and we thus first 
explain the idea informally. 
Prover will use a simple balancing strategy, when starting with
$T_0\sim U_0$:
\begin{itemize}
	\item		
In each round Prover chooses $k=M_1$ (in the point $1$ of
the game), where $M_1$ is a sufficiently large constant computed from
the grammar $\calG$; she also uses only $\calB_j\subseteq \sim$, thus
keeping $T_i\sim U_i$ for all $i$.
\item
Prover balances (in the point $3$ of the second game-version),
e.g. by replacing $(T'_i,U'_i)=(G\sigma,U'_i)$ with  
$(T_{i+1},U_{i+1})=(G\sigma',U'_i)$ in the case of left-balancing,
\emph{only when} the respective \emph{special head $G$ 
has a bounded height}, bounded by some sufficiently large $M'_0$;
the above
$M_1$ was chosen sufficiently larger than $M'_0$.
\item
Obeying the above ``bounded-head'' condition, 
Prover balances in any round in which she has an opportunity,
but she has still another
constraint: \emph{Prover does not ``switch'' balancing sides
	in two consecutive
rounds}, i.e.,
if she does a left-balancing
in the $(i{+}1)$-th round, then she cannot do a right-balancing 
in the $(i{+}2)$-th round, and vice versa.
\end{itemize}
To sketch the idea why this strategy enables to present 
an infinite subsequence of $(T_1,U_1), (T_2,U_2),\dots$ as
an $(n,g)$-sequence, we first explore the case 
where Prover does a left-balancing in the $(i{+}1)$-th round.
Hence we have 
\begin{center}
$W=U_i\closelh{M'_0}_{M_1}
(T_{i+1},U_{i+1})=(G\sigma',U_{i+1})$
\end{center}
(using the notation in~(\ref{eq:closelh})), where $W$ is the
respective pivot and $(G\sigma',U_{i+1})$ the respective bal-result;
this is also illustrated in Fig.~\ref{fig:tworounds}.

We have two possibilities for the following $(i{+}2)$-th round:
\begin{enumerate}
	\item		
There is a left-balancing in the  $(i{+}2)$-th round.
\\
The pivot of this balancing 
is $W'=U_{i+1}$, and we have 
$W\gt{u}W'$ where
$|u|=M_1$; hence $W'$ is ``boundedly reachable'' from $W$ in this case. 

\item
Left-balancing (with a bounded head) is not possible
in the  $(i{+}2)$-th round.
\\
Here we have ``no change'', i.e., 
$(T_{i+2},U_{i+2})=(T'_{i+1}, U'_{i+1})$, and we will derive that 
$W\close_{2M_1}(T_{i+2},U_{i+2})$ 
(using the notation in~(\ref{eq:relclose})).
Now the pivot $W'$ of the first next balancing in future
will be again reachable from $W$; maybe not boundedly reachable from $W$ but
boundedly
reachable from a subterm of $W$.
\end{enumerate}
The claims in the case $2$ 
are based on the fact that the impossibility to do a
left-balancing in the  $(i{+}2)$-th round
entails that 
 the respective path 
$T_{i+1}=G\sigma'\gt{u}T'_{i+1}$ 
is a shortest path from $G\sigma'$ to $T'_{i+1}$ and
is steadily ``sinking'' (or
``popping''), exposing deeper and deeper subterms of $G\sigma'$; 
our choice of $M'_0$ and $M_1$ will guarantee that
$T_{i+1}=G\sigma'\gt{u}T'_{i+1}$ 
can be then written $G\sigma'\gt{u'}x_\ell\sigma'\gt{u''}T'_{i+1}$, 
thus ``erasing'' $G$ 
and exposing some
$x_\ell\sigma'$ that is reachable from $W$ within $M_1$ steps;
this is depicted in Fig.~\ref{fig:tworounds}, where
$x_\ell\sigma'=V'$.

A simple analysis now shows that
if there is no repeat in the sequence 
$(T_0,U_0),(T_1,U_1),(T_2,U_2),\dots $, then
there must be infinitely many balancing rounds, with the respective 
pivots denoted 
$W_1, W_2,W_3,\dots$, 
while each concrete pivot can 
repeat only boundedly many times.
In the special ``pivot path'' $W_1\gt{w_1}W_2\gt{w_2}W_3\gt{w_3}\cdots$ which we
touched on (recall 
that $W_{j+1}$ is boundedly reachable from a subterm of $W_j$)
we then must have a deepest subterm $V_0$ of $W_1$,
visited in some segment 
$W_{j_0}\gt{w_{j_0}}W_{j_0+1}$, written as 
$W_{j_0}\gt{w'_{j_0}}V_0\gt{w''_{j_0}}W_{j_0+1}$,
such that the path 
$V_0\gt{w''_{j_0}}W_{j_0+1}\gt{w_{j_0+1}}W_{j_0+2}\gt{w_{j_0+2}}\cdots$
does not visit any subterm of $V_0$.
Then the sequence of
bal-results related to $W_{j_0+1},W_{j_0+2},W_{j_0+3},\dots$ 
can be presented as an $(n,g)$-sequence,
where $n,g$ are determined by $\calG$.

The proof of the next lemma just makes clear all relevant 
technical details.

\begin{lemma}\label{lem:forcingngseq}
There are $n_0,g_0$ ($n_0\in\Nat$, $g_0:\Nat_+\rightarrow\Nat_+$)
determined by (in fact, computable from) grammar
$\calG$ such that Prover can force for any initial $T_0\sim U_0$ that
she either wins or the sequence $(T_1,U_1), (T_2,U_2),\dots$ has an
infinite $(n_0,g_0)$-subsequence.
\end{lemma}	

\begin{proof}
We first introduce some technical notions related
to a given grammar $\calG=(\calN,\act,\calR)$; 
in our notation we
assume that $\arity(A)=m$ for all $A\in\calN$.

If $A(x_1,\dots,x_m)\gt{w}x_i$ in $\calL^\ltsrul_\calG$, then we call
$w\in\calR^*$ an  \emph{$(A,i)$-sink word}.
We assume that for each pair  
$A\in\calN$, $i\in[1,m]$ there is a fixed 
shortest $(A,i)$-sink word $w_{[A,i]}$, and we put 
\begin{center}
$M_0=1+\max\,\{\,|w_{[A,i]}|; A\in\calN, i\in
	[1,m]\,\}$.
\end{center}
The words $w_{[A,i]}$ can be found and 
$M_0$
can be computed by a standard dynamic programming
approach.

\begin{quote}
{\small
We note that  $|w_{[A,i]}|=1$ if there is a rule
$A(x_1,\dots,x_m)\gt{a}x_i$; otherwise
 $|w_{[A,i]}|=1+|u|$ where $u$ is a shortest word such that 
 $E\gt{u}x_i$ for a rule $A(x_1,\dots,x_m)\gt{a}E$; moreover, the path
 $E\gt{u}x_i$ ``sinks'' along a branch in $E$ till a leaf $x_i$, and
 can be composed from the relevant shorter words
$w_{[B,j]}$. Though $M_0$ can be exponential (as demonstrated by the rules
$A_k(x_1)\gt{a}A_{k-1}(A_{k-1}(x_1)), 
A_{k-1}(x_1)\gt{a}A_{k-2}(A_{k-2}(x_1)),\cdots, 
A_2(x_1)\gt{a}A_{1}(A_{1}(x_1)), A_1(x_1)\gt{a}x_1$), it can be
computed in polynomial time.

If we find that there are no $(A,i)$-sink words for some $A,i$, then 
the $i$-th root-successor of any $A(G_1,\dots,G_m)$ plays no role
(i.e., its replacing does not change the equivalence class);  we
can then simply omit the $i$-th root-successors when the root is $A$.
We can thus 
decrease
$\arity(A)$, and modify the grammar rules accordingly
so that the LTSs  $\calL^\ltsrul_\calG$ and $\calL^\ltsact_\calG$ 
do not change, up
to isomorphism. We can thus indeed safely assume that 
there are $w_{[A,i]}$ for all $A\in\calN$,
$i\in[1,\arity(A)]$.
}
\end{quote}
A \emph{path} $V\gt{u}$ in $\calL^\ltsrul_\calG$
is \emph{root-performable},
if $A(x_1,\dots,x_m)\gt{u}$ where 
 $A=\termroot(V)$
(in which case  $u$ is enabled by any term with the root $A$).
A \emph{path} $V\gt{w}$ in $\calL^\ltsrul_\calG$
is a \emph{non-sink segment}, a \emph{non-sink} for short,
if $|w|=M_0$ and $V\gt{w}$ is
root-performable.
(For each root-successor $V'$ in $V$ we thus have
$V\gt{v}V'$ for some $v$ shorter than $w$.)

A \emph{path} $T\gt{u}T'$  in $\calL^\ltsrul_\calG$
is \emph{sinking} if it contains no non-sink,
i.e., for any partition $u=u_1u_2u_3$ with $|u_2|=M_0$ we have 
$u_2=u'_{2}u''_{2}$ ($u'_2\neq\varepsilon$) where
$T\gt{u_1}V'\gt{u'_{2}}V''\gt{u''_{2}u_3}T'$ and $V''$ is a
root-successor in $V'$.
Hence if $T\gt{u}T'$ is sinking, then 
 it can be written $T\gt{u_1}V\gt{u_2}T'$ where $|u_2|<M_0$ and
 $V$ is a
subterm of $T$ in depth at least $|u|\div M_0$.
(By $\div$ we denote integer division.)

Finally we consider a \emph{shortest path}
$T\gt{u}T'$ \emph{from $T$ to $T'$ that is not sinking}.
It can be written 
$T\gt{u_1}V\gt{u_2}V'\gt{u_3}T'$ where 
$V\gt{u_2}V'$ is the last non-sink. 
We can easily check that then $T'=G\sigma$ where 
$\depth(G)\leq M'_0$, for some $M'_0$ determined by $\calG$,
and $\range(\sigma)$ consists of the root-successors in $V$.

\begin{quote}
	{\small	
To verify the claim, we first note 
that we cannot have $T\gt{u_1}V\gt{u_2}V'\gt{u_{31}}V''\gt{u_{32}}T'$
where $V''$ is a root-successor in $V$, since there would be a shorter
path $T\gt{u_1}V\gt{w}V''\gt{u_{32}}T'$ from $T$ to $T'$ (for $w$
being
the relevant sink-word $w_{[A,i]}$, satisfying $|w|<M_0$).
Hence $V\gt{u_2u_3}T'$  is root-performable, and we have
$V=(A(x_1,\dots,x_m))\sigma\gt{u_2u_3}T'=G\sigma$ where
$A(x_1,\dots,x_m)\gt{u_2u_3}G$ and
$\range(\sigma)$ consists of the root-successors in $V$.
Since we took the last non-sink in $T\gt{u}T'$,
the path $V\gt{u_{2}u_3}T'$, and thus also 
$A(x_1,\dots,x_m)\gt{u_2u_3}G$,
is sinking after its first
step. Therefore $G$ is reachable within less than $M_0$ steps from a
subterm of the rhs $E$ of a rule $A(x_1,\dots,x_m)\gt{a}E$ in the set 
$\calR$ of rules in the grammar $\calG$. We can thus (generously) put
$M'_0=(1+M_0)\cdot\maxruleheight$ where 
$\maxruleheight=\max\,\{\,\depth(F)\mid $ there is a rule
$B(x_1,\dots,x_m)\gt{a}F$ in $\calR\}$.
}
\end{quote}
We now take $M_1$ such that $(M_1\div M_0)> M'_0$,
and show Prover's strategy in the $(i{+}1)$-th round, when starting 
with 
$T_i\sim U_i$:
		\begin{enumerate}[i/]
	\item
		Prover chooses $k=M_1$ 
		and
		$\{(T_i,U_i)\}=\calB_0\iscov\calB_1\iscov\calB_2\iscov\ldots\iscov\calB_{M_1}$
		where		
		$\calB_j\subseteq\sim$ (for all $j\in[1,M_1]$).
		Refuter chooses $(T'_i,U'_i)\in\calB_{M_1}$ and 
		we can fix some paths 
$T_i\gt{u_1}T'_i$, $U_i\gt{u_2}U'_i$
in $\calL^\ltsrul_\calG$,
where $|u_1|=|u_2|=M_1$, $\lab(u_1)=\lab(u_2)$ (recall again
Fig.~\ref{fig:roundsketch}).
	\item 
If $T_i\gt{u_1}T'_i$ is not a shortest path from $T_i$ to $T'_i$ or contains a non-sink,
and \emph{Prover did not do a right-balancing in the (previous) 
$i$-th round}, then she makes 
a left-balancing, 
replacing $(T'_i,U'_i)=(G\sigma,U'_i)$ with 
$(T_{i+1}, U_{i+1})=(G\sigma',U'_i)$, for some head $G$ with the
		smallest possible height. 
		(We know that $\depth(G)\leq M'_0$.)
	\item
If ii/ did not apply, and 
$U_i\gt{u_2}U'_i$ is not shortest or contains a non-sink,
and Prover did not do a left-balancing in the 
		$i$-th round, then she 
		makes 
a right-balancing, symmetrically to ii/. 
	\item
If none of ii/, iii/ applied, 
Prover 
puts $(T_{i+1},U_{i+1})=(T'_i,U'_i)$.
\end{enumerate}
Before analysing the outcome of the strategy, 
we recall that each
bal-result $(T_{i+1},U_{i+1})$ has its pivot $W$, where
$W\closelh{M'_0}_{M_1} (T_{i+1},U_{i+1})$ 
or $W\closerh{M'_0}_{M_1} (T_{i+1},U_{i+1})$
(recall the definition in~(\ref{eq:closelh})),
and we explore the case when Prover  does a left-balancing 
in the $(i{+}1)$-th round, with the pivot
$W=U_i$,
but she cannot do a left-balancing (and thus any balancing) in the 
$(i{+}2)$-th round, as depicted in Fig.~\ref{fig:tworounds}. (We omit the case with a right-balancing, since it
is symmetric.)

In the mentioned case we have
$W\closelh{M'_0}_{M_1} (T_{i+1},U_{i+1})= (G\sigma',U_{i+1})$, 
where $\depth(G)\leq M'_0$ and each $V'\in\range(\sigma')$ is reachable
from $W$ within $M_1$ steps.
Now 
the respective path $T_{i+1}=G\sigma'\gt{u}T'_{i+1}$ (created in the 
$(i{+}2)$-th round) is sinking (and
shortest). But then   
$W\close_{2M_1}(T_{i+2},U_{i+2})$ 
as can be easily verified.
\begin{quote}
	{\small	
Indeed, 
we have chosen $M_1$ large enough 
($(M_1\div M_0)> M'_0\geq\depth(G)$) so that the sinking path 
$G\sigma'\gt{u}T'_{i+1}$ can be written
$G\sigma'\gt{u'}x_\ell\sigma'\gt{u''}T'_{i+1}$
(where $G\gt{u'}x_\ell$); informally, the path 
sinks along a branch of $G$ until a leaf $x_\ell$
of $G$ (where $x_\ell\sigma'$ hangs). Since 
$x_\ell\sigma'$
is reachable from $W$ within $M_1$ steps, we have
$W\gt{v}x_\ell\sigma'\gt{u''}T'_{i+1}$ where $|vu''|\leq 2M_1$.
On the other hand, our definitions yield 
that $W=U_i\gt{v_1}U_{i+1}\gt{v_2}U_{i+2}$ for some words
$v_1,v_2$ where each has the length $M_1$. 
}
\end{quote}
We now
explore an infinite play from $T_0\sim U_0$ where Prover uses the
above strategy. We first note that
there are \emph{infinitely many balancings};
otherwise from some round on we would have constant 
sinking on both sides, which necessarily leads to a repeat since our
terms are regular.
\begin{quote}
	{\small	
Suppose we have 
$T_i\gt{u_i}T_{i+1}\gt{u_{i+1}}T_{i+2}\gt{u_{i+2}}\cdots$
where all paths $T_{i+j}\gt{u_{i+j}}T_{i+j+1}$ (each of length $M_1$)
are sinking. Then $T_{i+1}$ is reachable from a subterm of $T_i$ 
in less than $M_0$
steps, and we can thus write  $T_{i+1}=G\sigma$ where $\depth(G)\leq M'_0$
and all $V\in\range(\sigma)$ are subterms of $T_i$.
Then the path $T_{i+1}=G\sigma\gt{u_{i+1}}T_{i+2}$ first sinks along a
branch of $G$ until exposing a subterm of $T_i$; hence 
$T_{i+2}$ is also reachable from a subterm of $T_i$ 
in less than $M_0$ steps. Inductively we thus derive 
that each
$T_{i+j}$ is reachable from a subterm of $T_i$ in less than $M_0$
steps, hence all $T_i, T_{i+1}, T_{i+2}, \dots$ range over a finite
set. (Similarly  $U_i, U_{i+1}, U_{i+2}, \dots$
would range over a finite
set when there were only finitely many balancings.)
}
\end{quote}
We denote the pivots of our infinitely many balancings by
$W_1,W_2,W_3,\dots$, and we easily verify
that for each $j$ we have a path  
$W_j\gt{w_j}W_{j+1}$ (in $\calL^\ltsrul_\calG$)
of the form  
\begin{equation}\label{eq:consecpivots}
W_j\gt{v_1}V'\gt{v_2}V''\gt{v_3}W_{j+1}
\end{equation}
where $|v_1|, |v_3|$ are bounded (surely by $2M_1$) and 
$V''$ is a subterm of $V'$; though $v_2$ can be sometimes long (and
sometimes empty), we can
assume the path $V'\gt{v_2}V''$ to be sinking.

\begin{quote}
	{\small	
In the case of balancings in two consecutive rounds
(which are then both left-balancings, or both right-balancings), 
with pivots $W_j$
and $W_{j+1}$, we have $W_j\gt{u}W_{j+1}$ where $|u|=M_1$.
Suppose now two consecutive balancings, with pivots
$W_j\in\{T_{i_1},U_{i_1}\}$ and $W_{j+1}\in\{T_{i_2},U_{i_2}\}$
that did not happen in two consecutive rounds, hence $i_2\geq i_1{+}2$.
By our above analysis we have 
$W_j\close_{2M_1}(T_{i_1+2},U_{i_1+2})$, and the strategy implies 
that we have either
$T_{i_1+2}\gt{u_3}T_{i_1+3}\gt{u_4}T_{i_1+4}\gt{u_5}\cdots
\gt{u_{i_2}}T_{i_2}=W_{j+1}$
or $U_{i_1+2}\gt{u_3}U_{i_1+3}\gt{u_4}U_{i_1+4}\gt{u_5}\cdots
\gt{u_{i_2}}U_{i_2}=W_{j+1}$ where each (sub)path 
$T_{i_1+\ell-1}\gt{u_\ell}T_{i_1+\ell}$, or
$U_{i_1+\ell-1}\gt{u_\ell}U_{i_1+\ell}$,
has length $M_1$ and is sinking.
Similarly as previously, we derive that $W_{j+1}$ is reachable in
less than $M_0$ steps from a subterm of either $T_{i_1+2}$ or
$U_{i_1+2}$.
		}
	\end{quote}
Suppose now that the ``pivot path''
\begin{center}
$W_1\gt{w_1}W_{2}\gt{w_2}W_{3}\gt{w_3}\cdots$
\end{center} 
visits
subterms of $W_1$ infinitely often.
Then the pivots $W_j$ are
infinitely often boundedly reachable from a subterm of $W_1$,
as follows
from the form~(\ref{eq:consecpivots})
of paths $W_j\gt{w_j}W_{j+1}$. In this case one
pivot reappears infinitely often;
but there
are boundedly many bal-results for one pivot, and we would thus have a repeat.
\begin{quote}
	{\small	
Recall that the bal-result $(T,U)$ related to pivot $W$ satisfies
$W\closelh{M'_0}_{M_1} (T,U)$ or $W\closerh{M'_0}_{M_1} (T,U)$
(as defined in~(\ref{eq:closelh})); hence we have boundedly many
possible $(T,U)$ for one $W$.
		}
	\end{quote}
Some segment $W_{j_0}\gt{w_{j_0}}W_{j_0+1}$ thus visits a subterm
of $W_1$, denoted by $V_0$, for the last time. 
Hence $W_{j_0}\gt{w'_{j_0}}V_0\gt{w''_{j_0}}W_{j_0+1}$,
and 
the infinite path $V_0\gt{w''_{j_0}}\gt{w_{j_0+1}}\gt{w_{j_0+2}}\cdots$ is
root-performable; for $A=\termroot(V_0)$ we have
\begin{equation}\label{eq:Arootperform}
A(x_1,\dots,x_m)\gt{w''_{j_0}}G_1\gt{w_{j_0+1}}G_2\gt{w_{j_0+2}}G_3\gt{w_{j_0+3}}\cdots.
\end{equation}
Hence $V_0=(A(x_1,\dots,x_m))\sigma'$ and $W_{j_0+\ell}=G_\ell\sigma'$
($\ell=1,2,\dots$) for
$\sigma'$ whose range consists of the root-successors in $V_0$. 
We also note that
$\depth(G_\ell)$ can only boundedly grow (with
growing $\ell$).
\begin{quote}
	{\small	
		By the form of the paths~(\ref{eq:consecpivots}),
		we know that 
		$G_{\ell+1}$ is boundedly reachable from a
subterm of $G_{\ell}$; to be more precise, $G_{\ell+1}$
is reachable within $M_0$ steps from a subterm of a term that is
reachable within $2M_1$ steps from $G_\ell$. Hence we surely have
$\depth(G_{\ell+1})\leq \depth(G_{\ell})+ 
(2M_1{+}M_0)\cdot\maxruleheight$ where 
$\maxruleheight=\max\,\{\,\depth(F)\mid $ there is a rule
$B(x_1,\dots,x_m)\gt{a}F$ in $\calR\}$.
Since $\depth(G_{1})\leq 1+ 
(2M_1{+}M_0)\cdot\maxruleheight$, we have 
\\
$\depth(G_{\ell})\leq 1+ 
\ell\cdot(2M_1{+}M_0)\cdot\maxruleheight$.
}
\end{quote}
We are interested in the bal-results related to 
$W_{j_0+1}, W_{j_0+2},W_{j_0+3},\dots$, i.e.,
to $G_1\sigma'$, 
$G_2\sigma'$, $G_3\sigma'$, $\dots$. 
Since the bal-result
$(T,U)$
related to $G_\ell\sigma'$ satisfies 
$G_\ell\sigma'\closelh{M'_0}_{M_1} (T,U)$
or $G_\ell\sigma'\closerh{M'_0}_{M_1} (T,U)$,  
it is built from some finite bounded ``head-terms'', and some
``tail-terms'' that are subterms of $G_\ell\sigma'$ in depth at most
$M_1$. 
\begin{quote}
	{\small	
A path $W\gt{v}V$ obviously cannot ``expose'', i.e.  ``sink to'',
a subterm of $W$ that is deeper than $|v|$.
}
\end{quote}
It is useful to rather write  
\begin{equation}\label{eq:VAF}
V_0=(A(x_1,\dots,x_m))\sigma'=F\sigma
\end{equation}
for a finite term $F$ in which each branch has length $M_1$
if it is not a complete branch of $V_0$, and where $\range(\sigma)$ consists
of the subterms of $V_0$ with depth $M_1$. 
\begin{quote}
	{\small	
To get $F$ and $\sigma$, for each particular 
occurrence of a subterm $U$ of $V_0$ that has depth
$M_1$ in $V_0$ we do the following:
we replace this occurrence of $U$ with a fresh variable $x_i$ and
we put $\sigma(x_i)=U$. The resulting term $F$ is a finite term with
$\depth(F)\leq M_1$, and
$\support(\sigma)$ consists of at most $m^{M_1}$ variables,
where $m$ is the maximum arity of noterminals of the grammar $\calG$.
Putting 
\begin{center}
$n_0= m^{M_1}$,
\end{center}
we get 
$|\support(\sigma)|\leq n_0$.
}
	\end{quote}
Recalling~(\ref{eq:Arootperform}) and~(\ref{eq:VAF}), we have
	\begin{center}
		$F\gt{w''_{j_0}}H_1\gt{w_{j_0+1}}H_2\gt{w_{j_0+2}}H_3\gt{w_{j_0+3}}\cdots$
\end{center}
where 
$W_{j_0+\ell}=G_\ell\sigma'=H_\ell\sigma$,
$\depth(H_\ell)\leq \depth(G_\ell)+M_1$,
and 
each
occurrence of $x_i\in\support(\sigma)$ in $H_\ell$ has depth at least
$M_1$ (for $\ell=1,2,\dots$).
The bal-result $(T,U)$ related to $W_{j_0+\ell}=H_\ell\sigma$ (satisfying
$H_\ell\sigma\closelh{M'_0}_{M_1} (T,U)$
or $H_\ell\sigma\closerh{M'_0}_{M_1} (T,U)$) can be thus written
\begin{center}
$(T,U)=(E_\ell\sigma, F_\ell\sigma)$ 
\end{center}
for finite terms $E_\ell,F_\ell$ whose
height, and thus also size,
can only boundedly grow with growing $\ell$
(since $\depth(H_\ell)$ can only boundedly grow with growing $\ell$).

\begin{quote}
{\small	
	If $H_\ell\sigma\gt{v}V$ (in $\calL^\ltsrul_\calG$) 
	where $|v|\leq M_1$,
	then $H_\ell\gt{v}H'$ where $V=H'\sigma$, since 
	the subterm-occurrences with depth at least $M_1$ in 
	$H_\ell\sigma$ need at least $M_1$ steps for being exposed.
Moreover, $\depth(H')\leq \depth(H_{\ell})+ 
M_1\cdot\maxruleheight$ (where $\maxruleheight$ bounds the 
height-increase in one step).
	
Recall that $H_\ell\sigma\closelh{M'_0}_{M_1} (T,U)$ entails 
$H_\ell\sigma\close_{M_1}U$, and 
$T=G\sigma''$ where $\depth(G)\leq M'_0$ and 
$H_\ell\sigma\close_{M_1}V$ for each
$V\in\range(\sigma'')$.
Hence each term from the set  $\{U\}\cup\{V\mid V\in\range(\sigma'')\}$
can be written in the form
$E'\sigma$ for some $E'$ with 
$\depth(E')\leq \depth(H_{\ell})+ M_1\cdot\maxruleheight$.
Therefore we can write $U=F_\ell\sigma$ and $T=G\sigma'''\sigma$
where the height of $F_\ell$ and of each $E'\in\range(\sigma''')$
is bounded by $\depth(H_{\ell})+ M_1\cdot\maxruleheight$.
Finally we put $E_\ell=G\sigma'''$. Hence
$(T,U)=(E_\ell\sigma,F_\ell\sigma)$, and we surely have
$\pressize(E_\ell,F_\ell)\leq 2\cdot (m^\textsc{H})^2$
where $m$ is the maximal arity of nonterminals and
$\textsc{H}=M'_0+\depth(H_\ell)+M_1\cdot\maxruleheight$.
Since $\depth(H_\ell)\leq \depth(G_\ell)+M_1$ and
$\depth(G_{\ell})\leq 1+ 
\ell\cdot(2M_1{+}M_0)\cdot\maxruleheight$, we get
\begin{center}
a function $g_0$, determined by the grammar $\calG$, 
\end{center}
for which 
$\pressize(E_\ell,F_\ell)\leq g_0(\ell)$, for $\ell=1,2,3,\dots$.
}
\end{quote}
Hence  the bal-results related to the pivots
$W_{j_0+1}, W_{j_0+2},W_{j_0+3},\dots$, i.e.,
to $H_1\sigma$, 
$H_2\sigma$, $H_3\sigma$, $\dots$, can be presented as   an $(n_0,g_0)$-sequence
\begin{center}
$(E_1\sigma, F_1\sigma),(E_2\sigma, F_2\sigma),(E_3\sigma, F_3\sigma),
\dots$, 
\end{center}
where $n_0,g_0$ that are determined by the
grammar $\calG$. 
\qed
\end{proof}

\textbf{The lengths of eqlevel-decreasing $(n,g)$-sequences are
bounded.}
Before proving Lemma~\ref{lem:realboundng}
we show
some useful facts and convenient notions, assuming a grammar
$\calG=(\calN,\act,\calR)$.
We first recall that $\eqlevel(E,F)\leq \eqlevel(E\sigma,F\sigma)$,
and note: 

\begin{proposition}\label{prop:bisimgetequation}
If $\eqlevel(E, F)= k<e=
\eqlevel(E\sigma, F\sigma)$
(where $e\in\Nat\cup\{\omega\}$)
then 
there are $x_i\in\support(\sigma)$, $H\neq x_i$,
and $w\in\act^*$, where $|w|\leq k$, such that
$E\gt{w}x_i$, $F\gt{w}H$ or $E\gt{w}H$, $F\gt{w}x_i$, and
$x_i\sigma\sim_{e-k} H\sigma$.
\end{proposition}

\begin{proof}
We take 
$\{(E\sigma,F\sigma)\}=
\calB_0\iscov\calB_1\iscov\cdots\iscov\calB_{k+1}$,
so that $\leasteqlev(\calB_j)=e-j$ for all $j\in[0,k{+}1]$.
(Recall Prop.~\ref{prop:simplecovering}.)
When trying to mimic this sequence by replacing 
$\sigma$ with the empty-support substitution and aiming to create
$\{(E,F)\}=
\calB'_0\iscov\calB'_1\iscov\cdots\iscov\calB'_{k+1}$,
we must get $(x_i,H)$ or $(H,x_i)$ with $H\neq x_i$
in some $\calB'_j$ for $j\leq k$
(instead of the original pair $(x_i\sigma,H\sigma)$ or
$(H\sigma,x_i\sigma)$), 
since 
otherwise we would prove $E\sim_{k+1}F$.
Since $\eqlevel(x_i\sigma,H\sigma)\geq
\leasteqlev(\calB_j)$ for the respective $j\leq k$, we surely have
$x_i\sigma\sim_{e-k}H\sigma$.
\qed
\end{proof}

By $\{(x_i,H)\}$ we denote the substitution that (only) replaces $x_i$ with
$H$ (i.e., $x_i\{(x_i,H)\}=H$ and 
 $x_j\{(x_i,H)\}=x_j$ for $j\neq i$.
 Hence  
$\{(x_i,H)\}\sigma$
is the substitution $\sigma'$ satisfying $x_i\sigma'=H\sigma$ and
$x_j\sigma'=x_j\sigma$ for all $j\neq i$.
We note that the (``limit'' regular) term
\begin{equation}\label{eq:Hprime}
H'=H\{(x_i,H)\}\{(x_i,H)\}\cdots
\end{equation}
is well defined and satisfies $H'=H\{(x_i,H')\}$:
a graph presentation of $H'$ arises from
a graph presentation of $H$ by redirecting each arc leading to $x_i$ (if
there is any) towards the root. (We have $H'=H$ if $x_i$ does not
occur in $H$, or if $H=x_i$.) Hence also
$\pressize(H')\leq\pressize(H)$.
E.g., for the terms in Fig.~\ref{fig:basicterm} we have
$E_2=E_1\{(x_2,E_1)\}\{(x_2,E_1)\}\{(x_2,E_1)\}\cdots\cdots=
E_1\{(x_2,E_2)\}$.

By $\sigma\remxi$ we denote the substitution arising
from $\sigma$ by removing $x_i$ from the
support (if it is there), i.e., 
\begin{center}
$x_i\sigma\remxi=x_i \textnormal{ and } x_j\sigma\remxi=x_j\sigma
	\textnormal{ for } j\neq i.$
\end{center}
	If $H\neq x_i$, then $x_i$ does not occur in $H'$
	defined by~(\ref{eq:Hprime}); we then have
$H'\sigma=H'\sigma\remxi$, and this enables an inductive argument 
in the proof of Lemma~\ref{lem:realboundng}, based on stepwise
decreasing the substitution support (i.e., the number $n$ in
eqlevel-decreasing $(n,g)$-sequences).

Recalling that $\sigma\sim_k\sigma'$ iff $x_j\sigma\sim_k x_j\sigma'$
for all $x_j\in\var$, and referring
to $H'$ in~(\ref{eq:Hprime}), we also note the following fact
(which follows from the congruence properties,
by a  repeated use of Prop.~\ref{prop:congruence}(2)):

\begin{proposition}\label{prop:congrlimit}
If $x_i\sigma\sim_k H\sigma$ 
and $H\neq x_i$,
then $\sigma\sim_k \{(x_i,H')\}\sigma\remxi$. 	
\end{proposition}	

\begin{proof}
Assume $H\neq x_i$; hence $x_i$ does not occur in
$H'$, and we also recall that $H'=H\{(x_i,H')\}$, and $H'\sigma=H'\sigma\remxi$.

For $j\neq i$ we obviously have 
$x_j\sigma = x_j\{(x_i,H')\}\sigma = x_j\{(x_i,H')\}\sigma\remxi$.
Hence we will be done if we show that 
$\eqlevel(x_i\sigma,H\sigma)=\eqlevel(x_i\sigma,x_i\{(x_i,H')\}\sigma\remxi)$,
i.e., if we show that
\begin{equation}\label{eq:ELHHprime}
\eqlevel(x_i\sigma,H\sigma)=\eqlevel(x_i\sigma,H'\sigma).
\end{equation}
Let
$\eqlevel(H\sigma,H'\sigma)=e$; this can be also written 
$\eqlevel(H\sigma,H\{(x_i,H')\}\sigma)=e$.

If $e=\omega$, 
then~(\ref{eq:ELHHprime}) is clear.
If $e<\omega$, then 
\begin{center}
$e=
\eqlevel(H\sigma,H\{(x_i,H')\}\sigma)>
\eqlevel(\sigma,\{(x_i,H')\}\sigma)=\eqlevel(x_i\sigma,H'\sigma)$
\end{center}
(where the inequality ``$>$'' follows from
Prop.~\ref{prop:congruence}(2)).
Thus $\eqlevel(H\sigma,H'\sigma)> \eqlevel(x_i\sigma,H'\sigma)$,
and~(\ref{eq:ELHHprime}) follows
by Prop.~\ref{prop:basicreplace}.
\qed
\end{proof}

\begin{quote}
	{\small
		\emph{Remark.}
We discussed a possible decomposition approach after noting the
congruence properties (Prop.~\ref{prop:congruence}). Now
		Propositions~\ref{prop:bisimgetequation}
		and~\ref{prop:congrlimit} also suggest a certain
decomposition approach, as we now sketch. 

Suppose we have $T\sim U$. We can present $(T,U)$ as
$(E\sigma,F\sigma)$ in many ways. 

E.g., if $A=\termroot(T)$ and
$B=\termroot(U)$ then we can put $E=A(x_1,\dots,x_m)$, 
$F=B(x_{m+1},\dots,x_{2m})$ (assuming $\arity(A)=\arity(B)=m$), and
for $i\in[1,m]$ we define $x_i\sigma$ to be the $i$-th
root-successor in $T$
and $x_{i+m}\sigma$ to be the $i$-th
root-successor in $U$.

For $(T,U)=(E\sigma,F\sigma)$ where $T\sim U$ there
are
two possibilities:
\begin{enumerate}
	\item		
either $E\sim F$, in which case $E\sigma'\sim F\sigma'$ for any
$\sigma'$,
	\item
or $\eqlevel(E,F)=k<\omega$.
\end{enumerate}
In the case $2$ we must have  $x_i\in\support(\sigma)$ and $H\neq
x_i$, where $E\close_{k}H$ or  $F\close_{k}H$, 
such that 
\begin{center}
$\sigma\sim \{(x_i,H')\}\sigma\remxi$
\end{center}
(by Prop.~\ref{prop:bisimgetequation} and~\ref{prop:congrlimit}), and
thus $E'\sigma\remxi\sim F'\sigma\remxi$ where
\begin{center}
$E'=E\{(x_i,H')\}$ and $F'=F\{(x_i,H')\}$.
\end{center}

We can even bound the size of $H$, and thus of $H'$, as follows:
$\pressize(H')\leq \pressize(E,F)+k\cdot \sizeinc$,
where $\sizeinc$ is defined in~(\ref{eq:defsizeinc}). 

We also note that for any $E,F,\sigma,x_i,H$ where $H\neq x_i$
(and maybe $E\sigma\not\sim F\sigma$), we have 
\begin{center}
$\eqlevel(E\sigma,F\sigma)\geq
\leasteqlev(\{(x_i\sigma, H'\sigma), (E'\sigma\remxi,F'\sigma\remxi)
\})$,
\end{center}
which can lead to a decomposition if the pairs
$(x_i\sigma, H'\sigma)$, $(E'\sigma\remxi,F'\sigma\remxi)$ are somehow
``smaller'' than $(E\sigma,F\sigma)$.

Moreover, in the case  $E'\sigma\remxi\sim F'\sigma\remxi$ we can
continue in the same way as above:
\begin{enumerate}
	\item		
either $E'\sim F'$, in which case $E'\sigma'\sim F'\sigma'$ for any
$\sigma'$,
	\item
or $\eqlevel(E',F')=k'<\omega$.
\end{enumerate}
In the latter case we get some $x_j\in\support(\sigma\remxi)$ and some
$G$ such that $E'\close_{k'}G$ or  $F'\close_{k'}G$, and 
$E''((\sigma\remxi)\remxj)\sim F''((\sigma\remxi)\remxj)$ where
$E''=E'\{(x_j,G')\}$ and $F''=F'\{(x_j,G')\}$, for
$G'=G\{(x_j,G)\}\{(x_j,G)\}\cdots$.

Continuing this reasoning, we must come to the case $1$ after at most
$n$ iterations where $n=|\support(\sigma)|$, maybe with the
empty-support substitution in the end.

A problem is to define an adequate size of the pairs, to transform the above
observations into a sound algorithm based on the respective
decompositions. In the algorithm based on our Prover-Refuter game we
circumvent this problem; we use the above observations
for a (conditional, nondeterministic)
computation of a bound on eqlevel-decreasing $(n,g)$-sequences.
	}		
\end{quote}	
We still add a few technical notions, useful for proving
Lemma~\ref{lem:realboundng}. For our assumed grammar
$\calG=(\calN,\act,\calR)$ we put 
\begin{equation}\label{eq:defsizeinc}
	\sizeinc=\max\,\{\,\pressize(E)\mid E \text{ is the rhs of a
		rule in }
	\calG\,\}.
\end{equation}	
We note that $F\gt{w}G$ implies $\pressize(G)\leq \pressize(F)+|w|\cdot\sizeinc$.

For any set $\calB\subseteq\trees_\calN\times\trees_\calN$ we put 
\begin{center}
$\maxeqlev(\calB)=\max\{\eqlevel(E,F)\mid (E,F)\in\calB\}$,
\end{center}
stipulating $\max \emptyset=0$.
($\leasteqlev(\calB)$ has been already defined.)

For any $b\in\Nat$, we put
\begin{quote}
$\size_{\leq b}=\{(E,F)\mid  \pressize(E,F)\leq b\,\}$, and
\\
$\bmelb=\maxeqlev(\size_{\leq b}\,\cap\not\sim)$.
\end{quote}
We note that $\bmelb$ is always a \emph{finite} number.

For any $n\in\Nat$ and
$g:\Nat_+\rightarrow\Nat_+$ we define $\ell_{n,g}\in\Nat$
by the following recursive definition:
\begin{quote}
$\ell_{0,g}\ \ \ =\,1+\bmelgone$, and
\\
$\ell_{n+1,g}=1+\bmelgone+\ell_{n,g'}$ where 
\end{quote}
\begin{equation}\label{eq:gprime}
g'(j)=
g(1+\bmelgone+j)+2\cdot(g(1)+\bmelgone\cdot\sizeinc) \textnormal{ for
all } j\in\Nat_+.
\end{equation}

\begin{lemma}\label{lem:realboundng}
Any eqlevel-decreasing $(n,g)$-sequence has length at most 
$\ell_{n,g}$.
\end{lemma}
\begin{proof}
By induction on $n$.
Assume an eqlevel-decreasing $(n,g)$-sequence 
\begin{center}
$(E_1\sigma,F_1\sigma),(E_2\sigma,F_2\sigma), \dots, 
(E_\ell\sigma,F_\ell\sigma)$,
\end{center}
which also entails 
$E_1\sigma\not\sim F_1\sigma$ by our definition.
Since $\eqlevel(E_1,F_1)\leq
\eqlevel(E_1\sigma,F_1\sigma)$,
we have $E_1\not\sim F_1$; 
moreover, $\eqlevel(E_1,F_1)\leq\bmelgone$ since
$\pressize(E_1,F_1)\leq g(1)$.

If $n=0$, then $(E_1,F_1)=(E_1\sigma,F_1\sigma)$, and thus
\begin{center}
$\ell\leq 1{+}\eqlevel(E_1,F_1)\leq 1{+}\bmelgone=\ell_{0,g}$;
\end{center}
we also
have $\ell\leq \ell_{0,g}$ if
$\eqlevel(E_1,F_1)=\eqlevel(E_1\sigma,F_1\sigma)$.

If $\eqlevel(E_1,F_1)=k<e=\eqlevel(E_1\sigma,F_1\sigma)$, then 
\begin{center}
$\sigma\sim_{e-k}\{(x_i,H')\}\sigma\remxi$ 
\end{center}
for some $x_i\in\support(\sigma)$ and some $H'$ with 
$\pressize(H')\leq g(1)+k\cdot\sizeinc\leq g(1)+\bmelgone\cdot\sizeinc$; this can be easily
derived from Prop.~\ref{prop:bisimgetequation} 
and~\ref{prop:congrlimit}.

We now put (\emph{shift}) $\shift=1{+}\bmelgone$;
hence $\shift>\eqlevel(E_1,F_1)=k$, and thus
\begin{center}
$e{-}k>\eqlevel(E_{\shift+1}\sigma,F_{\shift+1}\sigma)>
\cdots >
\eqlevel(E_{\shift+(\ell-\shift)}\sigma,F_{\shift+(\ell-\shift)}\sigma)$.
\end{center}
For $j=1,2,\dots,\ell-\shift$
we define
\begin{center}
$(E'_j,F'_j)=(E_{\shift+j}\{(x_i,H')\}, F_{\shift+j}\{(x_i,H')\})$.
\end{center}
Since $\eqlevel(E'_j\sigma,E_{\shift+j}\sigma)=
\eqlevel(E_{\shift+j}\sigma,E_{\shift+j}\{(x_i,H')\}\sigma)\geq
e{-}k$, and similarly 
$\eqlevel(F'_j\sigma,F_{\shift+j}\sigma)\geq e{-}k$,
by using Prop.~\ref{prop:basicreplace} we get
\begin{center}
$\eqlevel(E'_j\sigma,F'_j\sigma)=
\eqlevel(E_{\shift+j}\sigma, F_{\shift+j}\sigma)$.
\end{center}
Since $(E'_j\sigma,F'_j\sigma)=(E'_j\sigma\remxi,F'_j\sigma\remxi)$,
we get 
\begin{center}
$e{-}k>\eqlevel(E'_{1}\sigma\remxi,F'_{1}\sigma\remxi)>
\cdots >
\eqlevel(E'_{\ell-\shift}\sigma\remxi,F'_{\ell-\shift}\sigma\remxi)$.
\end{center}
Finally we note that
\begin{center}
$\pressize(E'_j,F'_j)\leq\pressize(E_{j+\shift},F_{j+\shift})+2\cdot\pressize(H)
\leq$\\
$\leq g(j+\shift)+2\cdot(g(1)+\bmelgone\cdot\sizeinc)=g'(j)$,
\end{center}
for $g'$ defined by~(\ref{eq:gprime}).
Hence 
\begin{center}
$(E'_1\sigma\remxi,F'_1\sigma\remxi),
(E'_2\sigma\remxi,F'_2\sigma\remxi),\dots,
(E'_{\ell-\shift}\sigma\remxi,F'_{\ell-\shift}\sigma\remxi)$
\end{center}
is an eqlevel-decreasing $(n{-}1,g')$-sequence.
By the induction hypothesis we have  
$\ell{-}\shift\leq\ell_{n-1,g'}$, and thus
$\ell\leq
1+\bmelgone+\ell_{n-1,g'}=\ell_{n,g}$.
\qed\end{proof}

If $T_0\sim U_0$, then
Prover can force a potentially infinite 
$(n_0,g_0)$-sequence for certain $n_0,g_0$ determined by $\calG$
(by Lemma~\ref{lem:forcingngseq}). 
She could claim a win
after creating an $(n_0,g_0)$-sequence longer than $\ell_{n_0,g_0}$,
if she could demonstrate the value $\ell_{n_0,g_0}$.
(By Lemma~\ref{lem:realboundng} it would be then clear that Refuter does not use the least-eqlevel
strategy or $T_0\sim U_0$.)
Inspecting the above, we can verify that for computing  
$\ell_{n,g}$ for concrete $n,g$ it suffices to know 
$\size_{\leq \bigb}\,\cap\not\sim$ for a sufficiently large $\bigb\in\Nat$,
and the values $g(j)$ for $j$ from a 
sufficiently large initial segment $[1,\insegm]$ of $\Nat$.
We can capture this by the following inductive definition:
\begin{itemize}
	\item		
We say that $\bigb\in\Nat$ is
a \emph{sufficient size-bound for} $n\in\Nat$ and
$g:\Nat_+\rightarrow\Nat_+$ (i.e., for computing $\ell_{n,g}$)
if $g(1)\leq\bigb$ and in the case $n>0$ we also have that 
$\bigb$ is a sufficient size-bound for $n{-}1,g'$ where
$g'$ is defined by~(\ref{eq:gprime}).
\item
We say that $\insegm\in\Nat$ is
a \emph{sufficient segment-bound} for $n\in\Nat$ and
$g:\Nat_+\rightarrow\Nat_+$ (i.e., for computing $\ell_{n,g}$)
if
$\insegm\geq 1$
and in the case $n>0$ we have that
$\insegm\geq 1+\bmelgone+\insegm'$ where 
$\insegm'$ 
is a sufficient segment-bound for $n{-}1,g'$ where
$g'$ is defined by~(\ref{eq:gprime}).
\end{itemize}
Finally we note that Prover can, when given a grammar $\calG$,
present some $n_0\in\Nat$ and $g_0:\Nat_+\rightarrow\Nat_+$ 
(or just the values $g_0(1), g_0(2),\dots,g_0(\insegm)$ for some
$\insegm\in\Nat$) and perform
 the above recursive computation for $\ell_{n_0,g_0}$, while 
 guessing a set 
 $\calC\subseteq(\size_{\leq\bigb}\,\cap\not\sim)$ for some $B\in\Nat$
 that is sufficient for this computation. We note that Prover can
 demonstrate that $\calC\subseteq\not\sim$, and also compute the
 eq-level for each pair in $\calC$ (recall
 Prop.~\ref{prop:negatcase}). 
 For $\Rest=\size_{\leq\bigb}\smallsetminus\noneq$ 
 Prover just claims that 
it is a subset of $\sim$, in which case her computation of
$\ell_{n_0,g_0}$ would be indeed correct; in reality she computes
a value $\ell^{\upC}_{n_0,g_0}$ that is dependent on her choice of
$\calC$.
We let Refuter to challenge 
the assumption $\Rest\subseteq\sim$, 
by choosing a pair from $\Rest$,
so that his least-eqlevel
strategy will still be winning if Prover does not guess $\calC$ 
correctly.
This idea will be
now formalized and embodied in the final game-version.

For $\calC\subseteq\not\sim$ and $b\in\Nat$ we put
\begin{center}
$\melb=\maxeqlev(\size_{\leq b}\,\cap\calC)$.
\end{center}
For triples $(\calC,n,g)$ where
$\calC\subseteq\not\sim$, $n\in\Nat$, $g:\Nat_+\rightarrow\Nat_+$
we define 
$\ell^{\upC}_{n,g}$ by the following recursive definition:
\begin{quote}
	$\ell^{\upC}_{0,g}\ \ \ =\,1+\melgone$, and
\\
$\ell^{\upC}_{n+1,g}=1+\melgone+\ell_{n,g'}$ where 
\end{quote}
\begin{equation}\label{eq:gprimerelC}
g'(j)=
g(1+\melgone+j)+2\cdot(g(1)+\melgone\cdot\sizeinc) \textnormal{ for
all } j\in\Nat_+.
\end{equation}

\begin{itemize}
	\item		
We say that 
		$\bigb\in\Nat$ is
a \emph{sufficient size-bound for}  a triple $(\calC,n,g)$ as above
if $\calC\subseteq \size_{\leq\bigb}$, $g(1)\leq\bigb$,
and in the case  
$n>0$ we also have that
$\bigb$ is a sufficient size-bound for $(\calC,n{-}1,g')$ where
$g'$ is defined by~(\ref{eq:gprimerelC}).
\item
We say that $\insegm\in\Nat$ is
a \emph{sufficient segment-bound} for 
$(\calC,n,g)$
if
$\insegm\geq 1$, and in the case $n>0$ we have that
$\insegm\geq 1+\melgone+\insegm'$ where 
$\insegm'$ 
is a sufficient segment-bound for $(\calC,n{-}1,g')$ where
$g'$ is defined by~(\ref{eq:gprimerelC}).
\end{itemize}

We now derive an analogy of Lemma~\ref{lem:realboundng}:

\begin{lemma}\label{lem:boundng}
Let $\bigb$ be a sufficient size-bound for $(\calC,n,g)$, 
where $\calC\subseteq\not\sim$, $n\in\Nat$,
$g:\Nat_+\rightarrow\Nat_+$,
and let 
$\Rest=\size_{\bigb}\smallsetminus\calC$.
Then any
eqlevel-decreasing $(n,g)$-sequence 
starting with a pair whose eq-level is less than 
$\leasteqlev(\Rest)$ has length at most $\ell^{\upC}_{n,g}$.
\end{lemma}

\begin{proof}
Let the assumption hold, and 
let us have an $(n,g)$-sequence 
\begin{center}
$(E_1\sigma,F_1\sigma),(E_2\sigma,F_2\sigma), \dots, 
(E_\ell\sigma,F_\ell\sigma)$ 
where
$\eqlevel(E_1\sigma,F_1\sigma)<\leasteqlev(\Rest)$.
\end{center}
Since $\eqlevel(E_1,F_1)\leq\eqlevel(E_1\sigma,F_1\sigma)$, 
and $\pressize(E_1,F_1)\leq g(1)\leq B$
(and thus $(E_1,F_1)\in \size_{\leq\bigb}$ and
$\eqlevel(E_1,F_1)<\leasteqlev(\Rest)$, which entails
$(E_1,F_1)\not\in\Rest$),
we must have
$(E_1,F_1)\in \calC$; therefore
$\eqlevel(E_1,F_1)\leq\melgone$.

If $n=0$, or more generally 
if $\eqlevel(E_1,F_1)=\eqlevel(E_1\sigma,F_1\sigma)$, 
then 
\begin{center}
$\ell\leq 1{+}\eqlevel(E_1,F_1)\leq 1{+}\melgone=\ell^{\upC}_{0,g}$.
\end{center}
If ($n>0$ and) $\eqlevel(E_1,F_1)=k<e=\eqlevel(E_1\sigma,F_1\sigma)$, then 
\begin{center}
$\sigma\sim_{e-k}\{(x_i,H')\}\sigma\remxi$ 
\end{center}
for some $x_i\in\support(\sigma)$ and some $H'$ with 
$\pressize(H')\leq g(1)+k\cdot\sizeinc\leq g(1)+\melgone\cdot\sizeinc$
(by Prop.~\ref{prop:bisimgetequation} and~\ref{prop:congrlimit}).

We now put (\emph{shift}) $\shift=1{+}\melgone$;
hence $\shift>\eqlevel(E_1,F_1)=k$, and thus
\begin{center}
$e{-}k>\eqlevel(E_{\shift+1}\sigma,F_{\shift+1}\sigma)>
\cdots >
\eqlevel(E_{\shift+(\ell-\shift)}\sigma,F_{\shift+(\ell-\shift)}\sigma)$.
\end{center}
For $j=1,2,\dots,\ell-\shift$
we define
$(E'_j,F'_j)=(E_{\shift+j}\{(x_i,H')\}, F_{\shift+j}\{(x_i,H')\})$;
we thus have  $\pressize(E'_j,F'_j)\leq g'(j)$ for
 $g'$ defined by~(\ref{eq:gprimerelC}).
We have 
$(E'_j\sigma,F'_j\sigma)=(E'_j\sigma\remxi,F'_j\sigma\remxi)$, and
by using Prop.~\ref{prop:basicreplace} we also derive
\begin{center}
$\eqlevel(E'_j\sigma\remxi,F'_j\sigma\remxi)=
\eqlevel(E_{\shift+j}\sigma, F_{\shift+j}\sigma)$.
\end{center}
Hence the sequence
\begin{center}
$(E'_1\sigma\remxi,F'_1\sigma\remxi),(E'_2\sigma\remxi,F'_2\sigma\remxi),\dots,
(E'_{\ell-\shift}\sigma\remxi,F'_{\ell-\shift}\sigma\remxi)$
\end{center}
is an eqlevel-decreasing $(n{-}1,g')$-sequence,
and $\bigb$ is sufficient for $(\calC,n{-}1,g')$ (since $\bigb$ is assumed 
sufficient for $(\calC,n,g)$).
The induction hypothesis thus implies that
$\ell{-}\shift\leq\ell^{\upC}_{n-1,g'}$,
and thus $\ell\leq
1+\melgone+\ell^{\upC}_{n-1,g'}=\ell^{\upC}_{n,g}$.
\qed
\end{proof}

\textbf{Prover-Refuter game (third version).}
We separate
$\calG$ from the initial pair, now denoted $(E_0,F_0)$,
to stress that the initial phase depends on $\calG$ only.

\begin{enumerate}[i)]
	\item
		A grammar $\calG=(\calN,\act,\calR)$ is given.
	\item 
Prover 
provides  
some finite set $\calC\subseteq\trees_\calN\times\trees_\calN$,
some $n_0\in\Nat$, a sequence of increasing values denoted $g_0(1),
g_0(2),\dots, g_0(\insegm)$ 
for some $\insegm\in\Nat$, and some $\bigb\in\Nat$ such that 
$\calC\subseteq\size_{\leq\bigb}$.
For each pair $(E,F)\in\calC$ Prover provides $e\in\Nat$ and 
demonstrates that $\eqlevel(E,F)=e$ (recall Prop.~\ref{prop:negatcase});
thus $\calC\subseteq\,\not\sim$.
Prover now computes $\ell^{\upC}_{n_0,g_0}$, using the recursive
definition given before~(\ref{eq:gprimerelC});
this fails when $\bigb$ or $\insegm$ are not sufficiently large.

\item 
An initial pair $(E_0,F_0)$ is given.
\item
For $\Rest=\size_{\leq\bigb}\smallsetminus\noneq$,
	Refuter chooses 
 $(T_0,U_0)$
from $\{(E_0,F_0)\}\cup\Rest$ 
(with the least eq-level when using
the least-eqlevel strategy).
\item 
Now a play of the second game-version starts with  $(T_0,U_0)$. 
A new feature is that Prover can claim her win when she shows that 
$(T_1,U_1), (T_2,U_2), \dots $ contains an $(n,g)$-subsequence that is
longer than $\ell^{\upC}_{n_0,g_0}$.
\end{enumerate}
The least-eqlevel strategy still guarantees
Refuter's win for $E_0\not\sim F_0$; Prover can never win by the new
game-rule (i.e., by exceeding $\ell^{\upC}_{n_0,g_0}$), due to
Lemma~\ref{lem:boundng}. On the other hand, Prover
can correctly guess 
$\calC=\size_{\leq \bigb}\,\cap\not\sim$
for $\bigb$ that is sufficient for computing (the real) $\ell_{n_0,g_0}$ 
(related to $n_0,g_0$ that are guaranteed for $\calG$ by 
Lemma~\ref{lem:forcingngseq}), and she can force her win
when $E_0\sim F_0$.

Since a winning strategy of Prover (for any $\calG, E_0, F_0$ where
$E_0\sim F_0$)
is finitely presentable and
effectively verifiable
(which easily follows from the fact that Refuter always has only
finitely many options when it is his turn),
a proof of Theorem~\ref{th:bisdecid} is now
clear.

\section{Additional Remarks}\label{sec:addrem}

Theorem~\ref{th:bisdecid} just states the existence of an algorithm
deciding \emph{bisimulation equivalence}
of first-order grammars, or, in more detail,
computing the respective eq-levels.
But the proof can be surely adapted to more general statements.
It would be a technical
exercise to phrase the proof in some more general terms, not referring
to bisimilarity. E.g., we could speak about some more general
(stratified) equivalence with a related 
notion of covering $\calB\iscov\calB'$ 
with some properties like those captured in 
Prop.~\ref{prop:simplecovering}, etc. 
As usual, a question in such cases is to what extent it makes good
sense. E.g., do we get new worthwhile decidability results in such a
way?

If we look at a (straightforward) transformation from pushdown
automata (PDA) to
FO-grammars (given here in Appendix for completeness), we note how
FO-grammars ``swallow'' deterministic popping $\varepsilon$-steps in PDA.
(If there is no other rule
for $A$ than $A(x_1,\dots,x_m)\gt{\varepsilon}x_i$, then any (sub)term 
$A(G_1,\dots,G_m)$ can be immediately replaced with $G_i$.)
A question posed by Stirling was if bisimilarity
of PDA with just popping 
$\varepsilon$-steps (where some nondeterminism is allowed) is still
decidable. 
This was answered negatively in~\cite{DBLP:journals/jacm/JancarS08}.

In our term-framework we extend the action set in
$\calG=(\calN,\act,\calR)$ with a \emph{silent action}, 
denoted $\varepsilon$, and we also allow 
\emph{$\varepsilon$-rules} 
 $A(x_1,\dots,x_m)\gt{\varepsilon}E$.
 The associated LTS is then
 $\calL^{\ltsact}_\calG=(\trees_\calN,\act\cup\{\varepsilon\},
 (\gt{a})_{a\in\act\cup\{\varepsilon\}})$ that naturally extends the
LTS $\calL^{\ltsact}_\calG$ defined for the case with no
$\varepsilon$-rules.
The \emph{collapsed LTS} 
 $\calL^{\ltsact{-}\coll}_\calG$ arises from
 $\calL^{\ltsact}_\calG$ by ``swallowing'' the
 $\varepsilon$-transitions, i.e., we have no $\varepsilon$-transitions
 in $\calL^{\ltsact{-}\coll}_\calG$, and
 $F\gt{a}H$ in $\calL^{\ltsact{-}\coll}_\calG$ if 
 in  $\calL^{\ltsact}_\calG$ 
 there is a path of the form
 \begin{equation}\label{eq:epsswallow}
F=F_0\gt{\varepsilon}F_1\gt{\varepsilon}\cdots
\gt{\varepsilon}F_{k_1}\gt{a}
H_0\gt{\varepsilon}H_1\gt{\varepsilon}\cdots
\gt{\varepsilon}H_{k_2}=H.
\end{equation}
A construction in~\cite{DBLP:journals/jacm/JancarS08} shows that 
bisimilarity in 
$\calL^{\ltsact{-}\coll}_\calG$ is undecidable, even
when all $\varepsilon$-rules are popping, i.e. of the form
$A(x_1,\dots,x_m)\gt{\varepsilon}x_i$.
As also noted in~\cite{DBLP:journals/jacm/JancarS08}, 
the same proof construction 
also shows that 
\emph{weak bisimilarity} (for PDA or for $\calL^{\ltsact}_\calG$
of FO-grammars where silent
popping moves are allowed) is undecidable.

\begin{quote}
	{\small
In fact, the construction for undecidability 
in~\cite{DBLP:journals/jacm/JancarS08} works also when we
do not include the silent ``post-transitions'', i.e., if we require
$k_2=0$ 
in~(\ref{eq:epsswallow}); thus the undecidability 
also holds for the respective equivalence
that is finer than
weak bisimilarity.
	}		
\end{quote}	
The undecidability results have been  recently refined, using
 branching bisimilarity~\cite{DBLP:journals/corr/YinFHHT14}.

\begin{quote}
{\small
In branching bisimilarity we also exclude the silent
``post-transitions'' (as mentioned above) but there is also a
``semantical'' constraint: the silent ``pre-transitions'' are supposed
to be not changing the equivalence-class.
Formally, given $\calG=(\calN,\act,\calR)$, where rules 
$A(x_1,\dots,x_m)\gt{\varepsilon}G$ are allowed,
we can define a \emph{branching bisimulation} 
in the (non-collapsed) LTS $\calL^{\ltsact}_\calG$
as a symmetric relation 
$\calB\subseteq\trees_\calN\times\trees_\calN$
where for each 
move $E\gt{a}E'$ in a pair $(E,F)\in\calB$, for
$a\in\act\cup\{\varepsilon\}$, 
there is a sequence, a \emph{response},
$F=F_0\gt{\varepsilon}F_1\gt{\varepsilon}\cdots
\gt{\varepsilon}F_k\gt{a}F'$ such that
$(E,F_i)\in\calB$ for all
$i\in[0,k]$, and $(E',F')\in\calB$; if $a=\varepsilon$, then 
it suffices that $(E',F)\in\calB$ (i.e., the response might be empty). 
}
\end{quote}
For ``pushing'' $\varepsilon$-rules the construction for
undecidability from~\cite{DBLP:journals/jacm/JancarS08} can be again
easily adapted to branching bisimilarity. But if we only allow 
popping  $\varepsilon$-rules, of the form 
$A(x_1,\dots,x_m)\gt{\varepsilon}x_i$, then the construction from
~\cite{DBLP:journals/jacm/JancarS08} cannot be used for branching
bisimilarity; there the silent
``pre-transitions'' do not keep the same equivalence class. 
This was noted by
 Y. Fu and Q. Yin~\cite{yuxi-pdadecid-14} who 
announced a result that
would translate in our setting as
the \emph{decidability of branching bisimilarity
of FO-grammars with popping $\varepsilon$-rules}.
\begin{quote}
{\small
In fact, ~\cite{yuxi-pdadecid-14} also announces the decidability for 
pushing $\varepsilon$-rules in the context of so called normed PDA
processes. The crucial idea is that the above ``responses''
$F=F_0\gt{\varepsilon}F_1\gt{\varepsilon}\cdots
\gt{\varepsilon}F_k\gt{a}F'$ to the moves $E\gt{a}E'$ can be bounded
in this context;
i.e., Prover gets again only boundedly many possibilities how to cover
a given finite set $\calB$ (in the adapted version of $\calB\iscov
\calB'$), which allows us to proceed essentially 
in the way that we used in this paper.
}
\end{quote}
Hence Y. Fu and Q. Yin have noticed that it indeed makes good 
sense to try to
adapt the decidability proof for bisimilarity to get further 
results. An adaptation of an existing proof seems necessary,
since the branching bisimilarity problem that
they study, in particular for PDA with popping $\varepsilon$-steps,
does not seem to be easily reducible 
to the known decidable (bisimilarity) problem; one has thus to 
delve into the existing proofs, looking for their possible adaptations.

Y. Fu and Q. Yin have chosen to build 
on Stirling's paper~\cite{stirling-pda-00}. They adapt Stirling's
tableau approach to a new model that they invented. In fact, 
when one looks at 
their
technical model, it seems clear that it could be smoothly replaced 
with the first-order-term framework used here 
(and already in~\cite{JancarLICS12}, i.e.,
in the paper of which the authors of~\cite{yuxi-pdadecid-14}
became aware only afterwards, as they say in their conclusions).
Analysing their procedure and its relation to our Prover-Refuter game
would require a nontrivial technical work; here 
we thus suggest 
a direct adaptation of the game that captures the announced result.

\subsubsection*{Adaptation of Prover-Refuter game.}
If we want to use our framework of the Prover-Refuter game
directly to branching bisimilarity
of FO-grammars with popping $\varepsilon$-rules, 
thus modifying the relation $\calB\iscov\calB'$ accordingly,
we encounter a
technical problem. Though in a pair $(E,F)$ each move 
$E\gt{a}E'$ ($a\in \act\cup\{\varepsilon\}$) still
has 
only finitely many possible responses 
$F=F_0\gt{\varepsilon}F_1\gt{\varepsilon}\cdots
\gt{\varepsilon}F_k\gt{a}F'$,
their number is not bounded (by a quantity determined by the
underlying grammar $\calG$), since the responses might
be sinking to subterms of $F$ in unbounded depths. 
This causes, e.g., that the bal-result related to a pivot might not be
``boundedly close'' to the pivot. But we can require that Prover
avoids such unbounded responses; she can always tell,
whenever she presents a new (sub)term $V$, 
if $V$ is equivalent with a root-successor $V'$ in $V$,
and she must then behave consistently with her claims;
we can imagine that she colours the respective arcs (form the root of
$V$ to the root of $V'$) as ``\emph{blue}''.
\begin{quote}
	{\small
	Recall that any response
$F=F_0\gt{\varepsilon}F_1\gt{\varepsilon}\cdots
\gt{\varepsilon}F_k\gt{a}F'$ should not change the equivalence class
when ``sinking'' from $F$ to $F_k$; this sinking can be done only along such
blue arcs if Prover colours the arcs correctly.
}
\end{quote}
Such a blue arc, i.e. a claim that $V\sim V'$ where $V'$ is a root-successor in $V$,
can be also challenged by Refuter, 
 but it can be used (later)
for replacing $V$ with $V'$ when this should not affect
the current eq-level,
if Refuter uses the least-eqlevel strategy.
In this way the respective possible transitions
$V\gt{\varepsilon}V'$ (where $V\sim
V'$) are
also ``swallowed'', similarly as deterministic popping
$\varepsilon$-steps, which are swallowed ``automatically''. 
\begin{quote}
	{\small
Recall that if there is no other rule
for $A$ than $A(x_1,\dots,x_m)\gt{\varepsilon}x_i$, then any (sub)term 
$A(G_1,\dots,G_m)$ can be immediately replaced with $G_i$.
}
\end{quote}
We thus recover the ``bounded-closeness'' properties, and we can
accordingly adapt the proof that was used in the case with no
$\varepsilon$-steps. 

Below we suggest a possible way how to formalize the above idea
 of ``blue arcs''.
We stay in the framework of bisimilarity of FO-grammars (with
no $\varepsilon$-rules),
but the decidability for branching bisimilarity
(of FO-grammars with popping $\varepsilon$-steps)
follows routinely  after this adaptation.
We thus consider a grammar $\calG=(\calN,\act,\calR)$ and the
Prover-Refuter game as they were defined previously.

\textbf{Quotient graph-presentations, and related decompositions}.
\\
We now stress more explicitly that
we deal with finite objects in the Prover-Refuter game, i.e., with 
graph-presentations (of terms),
called just \emph{graphs} in what follows, rather than 
with the terms themselves. We
use  symbols  $\pE, \pF,\dots,\pT,\pU,\dots$ to range over graphs
(while $E, F,\dots,T,U,\dots$ range over terms).
\begin{quote}
	{\small
 In fact, we can take graphs as the states in
the LTSs $\calL^\ltsrul_\calG$,
$\calL^\ltsact_\calG$;
we made clear how the transitions look like in this case.
Nevertheless, our aim is to convey the main idea, not 
to delve into unnecessary technicalities.
}
\end{quote}
Each graph $\pV$ has finitely many nodes,
and 
each node $\gnode$ corresponds to the term $\term(\gnode)$ rooted in
$\gnode$. 

We write shortly $\gnode_1\sim_k\gnode_2$
instead of $\term(\gnode_1)\sim_k\term(\gnode_2)$; 
similarly
we write  $\eqlevel(\gnode_1,\gnode_2)$
instead of $\eqlevel(\term(\gnode_1),\term(\gnode_2))$.

By $\area(\gnode)$ for a node $\gnode$ of $\pV$ we mean the
restriction of $\pV$ to the nodes occurring on (directed) paths in
$\pV$ that start in $\gnode$. 
\begin{quote}
	{\small
We thus have a correspondence (not necessarily one-to-one)
between 
the nodes in $\area(\gnode)$ and the subterms of $\term(\gnode)$.
}
\end{quote}
Let us now consider a graph $\pV$ with a partition $\calP$ of
its nodes. By $\pV\partit$ we denote a (chosen)
\emph{quotient of} $\pV$ \emph{w.r.t.
$\calP$} arising as follows: In each partition-class of $\calP$ 
we choose a
representant-node; the nodes
of $\pV\partit$ are the representant-nodes of all partition
classes, and if an arc from a representant leads to a
non-representant $\graphnode$, then we redirect it to the
representant $\graphnode'$ of the partition-class containing
$\graphnode$.
We thus also get a mapping $\redP$, ``\emph{reducing}'' each node
$\gnode$ of $\pV$ to $\redP(\gnode)$, which is the node in
$\pV\partit$ representing the partition-class of $\gnode$.

For a node $\graphnode$ of $\pV$, by $\redPone(\gnode)$
(``depth-$1$-reducing of $\gnode$'') we mean (a copy of)
the node $\gnode$ in the
graph arising as follows: we take a disjoint union of $\pV$ and
$\pV\partit$, where  
each outgoing arc of $\gnode$ in $\pV$, leading to 
some $\gnode'$, is redirected to the node $\redP(\gnode')$ in 
$\pV\partit$. Thus the nodes in $\area(\redPone(\gnode))$ are taken
from $\pV\partit$, with the exception of the root. 

We now define the \emph{decomposition of $\pV$ by $\calP$}:
\begin{center}
$\decP(\pV)=\{(\redPone(\gnode),\redP(\gnode))\mid \gnode$ is a
node in $\pV\}$.
\end{center}
Hence $\decP(\pV)$ is a set of pairs of nodes in a graph; the graph
arises from
$\pV\partit$ by adding copies of the nodes from $\pV$ whose outgoing
arcs are directed into $\pV\partit$.

\begin{proposition}\label{prop:declowerEL}
Let $\pV$ be a graph and $\calP$ a partition of its nodes.
If $\gnode_1,\gnode_2$ are in the same
partition-class of $\calP$, then
$\eqlevel(\gnode_1,\gnode_2)\geq\leasteqlev(\decP(\pV))$.
\end{proposition}	

\begin{proof}
Suppose some $\pV$, $\calP$; let  $\leasteqlev(\decP(\pV))=e$.
By the congruence properties 
we derive for each node $\gnode$ in $\pV$ that 
\begin{center}
$\eqlevel(\gnode,\redP(\gnode))
\geq e$
and $\eqlevel(\gnode,\redPone(\gnode))
\geq e$.
\end{center} 
\begin{quote}
	{\small	
We can show this as follows.
Let $e'=\leasteqlev(\{(\gnode,\redP(\gnode))\mid \gnode$ is a
 node in $\pV\})$. 
 Then $\gnode\sim_{e'+1}\redPone(\gnode)$, 
 by Prop.~\ref{prop:congruence}(2). We thus 
have
$\gnode\sim_{e'}\redP(\gnode)\sim_{e}\redPone(\gnode)\sim_{e'+1}\gnode$,
which implies $e'\geq e$.
}
\end{quote}
For $\gnode_1,\gnode_2$ where $\redP(\gnode_1)=\redP(\gnode_2)$ 
we thus have 
\begin{center}
$\gnode_1\sim_e \redPone(\gnode_1)\sim_e \redP(\gnode_1)
=\redP(\gnode_2)\sim_e \gnode_2$.
\end{center}
\qed
\end{proof}
We will particularly use the decompositions of
graphs $\pV$ that are induced by sets of arcs in $\pV$ 
(later called ``blue arcs'' or ``red arcs'').
Suppose $\pV$ and a set of its arcs; we call the arcs in the set
``blue''. This defines the least partition 
where the source-node and the target-node of any blue arc are in the
same partition-class.
\begin{quote}
	{\small	
A partition $\calP$ of the nodes of $\pV$
determines the set $\{(\gnode_1,\gnode_2)\mid \gnode_1,\gnode_2$ are
in the same partition-class of $\calP\}$. Hence partitions can be
naturally ordered by inclusion; we refer to this order when saying ``the
least partition such that ...''.
In the above case, two nodes are in the same partition-class iff there
is a ``blue-path'' between them in the \emph{undirected} graph version.
	}
\end{quote}	

\textbf{Modified Prover-Refuter game.}
Let us recall the second version of the game. We modify
it as follows.
\begin{enumerate}
	\item
We denote the initial pair $(\overline{T},\overline{U})$, and assume
that it is given 
by a graph $\pV$ with two designated nodes $\gnode'$, $\gnode''$
where 
$\term(\gnode')=\overline{T}$ and $\term(\gnode'')=\overline{U}$.

Prover now suggests a partition $\calP$ 
of the set of nodes of $\pV$ (generally, not necessarily by ``blue
arcs'') where
$\gnode', \gnode''$ must be in the same partition-class.
Now $\pV\partit$ is created, where each non-loop arc
is coloured \emph{black};
in this way Prover claims that $\gnode_1\not\sim\gnode_2$ for the
source-node $\gnode_1$ and the target-node $\gnode_2$.

\begin{quote}
	{\small
Any loop-arc from $\gnode$ to $\gnode$ trivially satisfies that its
source-node and its target-node are equivalent; 
we further ignore such
arcs in our discussion.
We also note that Prover could even demonstrate that $\gnode_1\not\sim\gnode_2$
but this is not
necessary here.
	}		
\end{quote}	
Prover claims 
that the above partition $\calP$ 
is induced by the bisimulation
equivalence; she thus also claims that $\leasteqlev(\decP(\pV))=\omega$.

Refuter now chooses a pair 
$(\redPone(\gnode),\redP(\gnode))$
from  $\decP(\pV)$, corresponding to a pair
$(T_0,U_0)$ of terms.
\begin{quote}
	{\small
If Refuter uses the least-eqlevel strategy, we have 
$\eqlevel(T_0,U_0)\leq\eqlevel(\overline{T},\overline{U})$
(by Prop.~\ref{prop:declowerEL}).
	}		
\end{quote}
The pair $(T_0,U_0)$ is thus, in fact, given by a pair
$(\gnode_{01},\gnode_{02})$ of nodes of a graph $\pV_0$ where 
only the outgoing arcs of $\gnode_{01}$  might be not
black (when ignoring the loop-arcs). Prover is supposed to colour each
outgoing arc of $\gnode_{01}$
as ``\emph{blue}'' iff its source-node and its target-node
are bisimilar.

\item
Prover will use a strategy (corresponding to the strategy
in the proof of Lemma~\ref{lem:forcingngseq}) that also guarantees  
that the $(i{+}1)$-th round starts with
a pair $(T_i,U_i)$ given by two nodes $\gnode_{i1},\gnode_{i2}$
in some graph $\pV_i$ where 
the arcs are coloured black or blue (ignoring the loop-arcs), 
and where a cycle in $\pV_i$ never contains a blue arc, and each blue arc is 
in a bounded distance (depth) from $\gnode_{i1}$ or $\gnode_{i2}$.
\begin{quote}
	{\small
As previously, by a ``bounded'' depth we mean that the respective
bound is determined by the underlying grammar $\calG$.
	}		
\end{quote}
\item
Whenever Prover presents a new graph (in the sets $\calB_j$), she is
supposed to colour each arc, from $\gnode$ to $\gnode'$, \emph{blue} 
if $\gnode\sim\gnode'$; by black arcs she claims non-equivalence.
She must be consistent with her previous choices.

Now the sets $\calB_j$ contain graphs with two designated nodes, and
with coloured arcs.
Refuter thus chooses $(T'_i,U'_i)$ by choosing a graph $\pV'_i$
with two designated nodes $\gnode'_{i1},\gnode'_{i2}$.

\item
If Prover does not make a balancing step, 
in the $(i{+}1)$-th round after Refuter has chosen $\pV'_i$
with  $\gnode'_{i1},\gnode'_{i2}$, then 
we define the partition
$\calP$ of the nodes in $\pV'_i$ as the least partition containing 
$(\gnode'_{i1},\gnode'_{i2})$ and the source-target pairs of all blue
arcs.
Refuter chooses a pair $(\redPone(\gnode),\redP(\gnode))$
from 
$\decP(\pV'_i)$, which presents the pair $(T_{i+1},U_{i+1})$.

\item 
Suppose that Prover makes a balancing step, say a left one, 
corresponding to replacing $(T'_i,U'_i)=(G\sigma,U'_i)$ with 
$(G\sigma',U'_i)$; this is naturally implemented 
in the graph, and we get a graph $\pV''_i$ instead of $\pV'_i$. 
In  $\pV''_i$ we recolour the blue arcs in $\area(\gnode'_{i2})$
(in the $U'_i$-area) to \emph{red}; we perform such a 
blue-to-red recolouring
also in the ``$\range(\sigma')$-area'', i.e., in $\area(\gnode)$
for  each
$\gnode$ that corresponds to the root of some $V'\in\range(\sigma')$.

Now we define the partition
$\calP$ of the nodes in $\pV''_i$ as the least partition containing 
$(\gnode'_{i1},\gnode'_{i2})$ and the source-target pairs of all
\emph{red} arcs.
Refuter chooses a pair $(\redPone(\gnode),\redP(\gnode))$
from 
$\decP(\pV''_i)$, which presents the pair $(T_{i+1},U_{i+1})$.
\begin{quote}
	{\small
		In the respective graph $\pV_{i+1}$ we have no red
		arcs but there can be the blue arcs inherited 
		from the special
		head $G$, which has a bounded height. 
	}		
\end{quote}

\end{enumerate}
When Prover uses the strategy described 
in the proof of Lemma~\ref{lem:forcingngseq}, while also guessing the
blue arcs correctly, she cannot lose when
$\overline{T}\sim\overline{U}$.  
It is a routine to verify that any bal-result 
is still ``boundedly-close'' to its respective pivot.
Now the pivots $W, W'$ of two consecutive balancings
(not necessarily in two consecutive rounds) might not satisfy that
$W'$ is boundedly reachable from a subterm of $W$, but $W'$ arises
from a term  boundedly reachable from a subterm of $W$ by some
replacings of bounded-depth subterms with other bounded-depth
subterms. This fact enables to derive Lemma~\ref{lem:forcingngseq}
as previously.

In the case with no $\varepsilon$-rules, the ``machinery'' of blue
arcs is not needed. But it makes sense when we consider branching
bisimilarity in the case of popping $\varepsilon$-rules. 
There the relation $\calB\iscov\calB'$ is modified appropriately, and
we allow Prover to use only responses 
$F=F_0\gt{\varepsilon}F_1\gt{\varepsilon}\cdots
\gt{\varepsilon}F_k\gt{a}F'$
where the ``pre-transitions'' 
can only sink 
along blue arcs, and thus always into bounded depths.
The decidability proof can be then finished  analogously to the 
case with no $\varepsilon$-rules.

\smallskip

\noindent
\textbf{Further remarks on related research.}
Further work is needed to fully understand the discussed problems.
E.g., even the case 
of BPA processes, generated by real-time PDA with
a single control-state, is not quite clear.
Here the bisimilarity
problem is EXPTIME-hard~\cite{Kiefer13} and in 2-EXPTIME~\cite{DBLP:conf/mfcs/BurkartCS95} 
(proven explicitly in~\cite{Jan12b}); for the subclass of normed BPA
the problem is polynomial~\cite{HiJeMo96}
(see~\cite{CzLa10} for the best published upper bound).

\bibliographystyle{splncs03}
\bibliography{root}

\subsection*{Appendix}

\textbf{A transformation of PDA  to first-order grammars.}
\\
By a \emph{pushdown automaton} (PDA) we mean a structure
$\calM=(Q,\Gamma,\act,\Delta)$ where $Q, \Gamma,\act$ are finite sets of \emph{control
states}, of \emph{stack symbols}, and of \emph{actions} (or
\emph{input letters}), respectively; $\Delta$ is a finite set
of \emph{pushdown-rules} of the form $pX\gt{a}q\alpha$ where $p,q\in
Q$, $X\in\Gamma$, $\alpha\in\Gamma^*$, and
$a\in\act\cup\{\varepsilon\}$.
By a \emph{configuration} we mean any string $q\beta\bot$ where $q\in Q$,
$\beta\in\Gamma^*$, and $\bot$ is a special \emph{bottom-of-the-stack
symbol} (where $\bot\not\in\Gamma$). 

A PDA $\calM=(Q,\Gamma,\act,\Delta)$ has the associated LTS
\begin{center}
$\calL_\calM=(\conf,\act\cup\{\varepsilon\},(\gt{a})_{a\in\act\cup\{\varepsilon\}})$
\end{center}
where $\conf$ is the set of configurations, and the transitions are
induced by the pushdown-rules as follows:
\begin{center}
if $pX\gt{a}q\alpha$ is in $\Delta$, then 
$pX\beta\bot\gt{a}q\alpha\beta\bot$ for any $\beta\in\Gamma^*$.
\end{center}

Suppose $Q=\{q_1,q_2,\dots,q_m\}$. Then 
a configuration $q_iY_1Y_2\dots Y_k\bot$ 
can be naturally viewed as the term $\calT(q_iY_1Y_2\dots Y_k\bot)$ defined
inductively by the points $1$ and $2$ below. 
Hence we view each pair $(q_i,Y)$ of a control state and a stack symbol as 
a nonterminal $[q_iY]$ with arity $m$; a special case is $\bot$ with arity $0$.
A pushdown rule $q_iY\gt{a}q_j\beta$ is rewritten to 
 $q_iYx\gt{a}q_j\beta x$ for a special formal symbol $x$, and
\begin{center}
 $q_iYx\gt{a}q_j\beta x$ 
 is transformed to  $\calT(q_iYx)\gt{a}\calT(q_j\beta x)$, 
\end{center}
where we also use the point $3$ below:
\begin{enumerate}
	\item
		$\calT(q_i\bot)=\bot$,
	\item
		$\calT(q_iY\alpha)=[q_iY](\calT(q_1\alpha),\calT(q_2\alpha),
\dots, \calT(q_m\alpha))$.
	\item
		$\calT(q_jx)=x_j$.
\end{enumerate}
Hence
$\calT(q_iYx)=[q_iY](x_1,x_2,\dots,x_m)$.
In fact, we can modify the operator $\calT$ 
for \emph{deterministic popping $\varepsilon$-rules}: 
If there is no other pushdown-rule for $q_i,Y$ than 
$q_iY\gt{\varepsilon}q_j$, then instead of creating 
the grammar rule
$\calT(q_iYx)\gt{\varepsilon}\calT(q_jx)$ 
we might modify the transformation $\calT$ by putting
$\calT(q_iY\alpha)=\calT(q_j\alpha)$; we have thus
``\emph{swallowed}'' the respective $\varepsilon$-step.
(The branches of the syntactic tree of 
 $\calT(q\alpha)$ can have varying lengths in this case.)

We thus do not need $\varepsilon$-rules in FO-grammars for expressing
PDA where only deterministic popping $\varepsilon$-moves are allowed.

\end{document}